\documentclass[a4paper,UKenglish]{lipics-v2016}
\pdfoutput=1
\usepackage{microtype}%if unwanted, comment out or use option "draft"

\title{Sequoidal Categories and Transfinite Games: A Coalgebraic Approach to Stateful Objects in Game Semantics\footnote{This work was partially supported by UK EPSRC Grant EP/K037633/1}}
\titlerunning{Sequoidal Categories and Transfinite Games} %optional, in case that the title is too long; the running title should fit into the top page column

\author[1]{William John Gowers}
\author[1]{James Laird}
\affil[1]{Department of Computer Science, University of Bath, Claverton Down, Bath.  BA2 7AY.  United Kingdom\\
  \texttt{W.J.Gowers@bath.ac.uk} \quad\texttt{jiml@cs.bath.ac.uk}}
\authorrunning{W.\,J. Gowers and J. Laird} %mandatory. First: Use abbreviated first/middle names. Second (only in severe cases): Use first author plus 'et. al.'

\Copyright{William John Gowers and James Laird}%mandatory, please use full first names. LIPIcs license is "CC-BY";  http://creativecommons.org/licenses/by/3.0/

\subjclass{F3.2.2 Denotational Semantics}% mandatory: Please choose ACM 1998 classifications from http://www.acm.org/about/class/ccs98-html . E.g., cite as "F.1.1 Models of Computation". 
\keywords{Game semantics, Stateful languages, Transfinite games, Sequoid operator}% mandatory: Please provide 1-5 keywords
\EventEditors{Filippo Bonchi and Barbara K\"onig}
\EventNoEds{2}
\EventLongTitle{7th Conference on Algebra and Coalgebra in Computer
Science (CALCO 2017)}
\EventShortTitle{CALCO 2017}
\EventAcronym{CALCO}
\EventYear{2017}
\EventDate{June 12--16, 2017}
\EventLocation{Ljubljana, Slovenia}
\EventLogo{}
\SeriesVolume{72}
\ArticleNo{13}
\usepackage{IEEEtrantools} % for \IEEEeqnarray
\usepackage{pbox} % for \pbox
\usepackage{cmll} % Linear logic symbols!
\usepackage{wasysym} % For left and right moons
\usepackage{mathtools}

\usepackage{tikz}
\usetikzlibrary{cd}
\usetikzlibrary{patterns}
\usetikzlibrary{calc}
\usetikzlibrary{decorations.pathmorphing}

\DeclareFontFamily{U}{mathb}{\hyphenchar\font45}
\DeclareFontShape{U}{mathb}{m}{n}{
      <5> <6> <7> <8> <9> <10> gen * mathb
      <10.95> mathb10 <12> <14.4> <17.28> <20.74> <24.88> mathb12
      }{}
\DeclareSymbolFont{mathb}{U}{mathb}{m}{n}

\DeclareMathSymbol{\sqsubsetneq}    {3}{mathb}{"88}
\DeclareMathSymbol{\varsqsubsetneq} {3}{mathb}{"8A}
\DeclareMathSymbol{\varsqsubsetneqq}{3}{mathb}{"92}
\DeclareMathSymbol{\sqsubsetneqq}   {3}{mathb}{"90}

\theoremstyle{plain}
\newtheorem{proposition}[theorem]{Proposition}
\newtheorem{exercise}{Exercise}
\newtheorem*{warning}{\lightning\ Warning \lightning}
\theoremstyle{definition}
\newtheorem{notation}[theorem]{Notation}

\newcommand*\from{\colon}
\def \inv {^{-1}}
\DeclareMathOperator{\id}{id}
\DeclareMathOperator{\pr}{pr}
\newcommand{\tensor}{\otimes}
\newcommand{\sequoid}{\oslash}
\renewcommand{\implies}{\multimap}
\newcommand{\comp}[2]{#2 ; #1}

\newcommand{\C}{\mathcal C}
\newcommand{\F}{\mathcal F}
\newcommand{\G}{\mathcal G}
\newcommand{\W}{\mathcal W}
\newcommand{\suchthat}{\,\colon\,}

\newcommand{\OP}{\{O,P\}}
\newcommand{\emptyplay}{\epsilon}
\newcommand{\prefix}{\sqsubseteq}
\newcommand{\pprefix}{\sqsubsetneqq}

\newcommand{\assoc}{{\mathtt{assoc}}}
\newcommand{\lunit}{{\mathtt{lunit}}}
\newcommand{\runit}{{\mathtt{runit}}}
\newcommand{\sym}{{\mathtt{sym}}}

\newcommand{\blank}{\,\underline{\hspace{1.5ex}}\,}
\newcommand{\der}{{\mathtt{der}}}

\newcommand{\wk}{{\mathtt{wk}}}
\newcommand{\toisom}{{\xrightarrow{\cong}}}
\newcommand{\passoc}{{\mathtt{passoc}}}

\newcommand{\run}{{\mathtt{r}}}

\newcommand{\fcoal}[1]{{\leftmoon #1 \rightmoon}}
\renewcommand{\subset}{\subseteq}

\newcommand{\dist}{{\mathtt{dist}}}
\newcommand{\dec}{{\mathtt{dec}}}
\renewcommand{\int}{{\mathtt{coh}}}

\DeclareMathOperator{\CCom}{CCom}
\newcommand{\catname}[1]{{\mathsf{#1}}}

\newcommand{\Ord}{\catname{Ord}}
\newcommand\oppcat[1]{#1^{\mathrm{op}}}
\newcommand{\bN}{{\mathbb{N}}}
\newcommand{\OK}{{\mathtt{OK}}}

\newcommand{\Var}{{\mathtt{Var}}}
\renewcommand{\read}{{\mathtt{read}}}
\newcommand{\wwrite}{{\mathtt{write}}}
\newcommand{\Rel}{\catname{Rel}}
\newcommand{\inr}{{\mathsf{inr}}}
\newcommand{\inl}{{\mathsf{inl}}}
\newcommand{{\Na}}{\bN}
\newcommand{{\cell}}{{\mathsf{cell}}}
\newcommand{\fix}{{\mathsf{fix}}}
\newcommand{\eq}{{\mathsf{eq}}}
\newcommand{\co}{{\textrm{\textexclamdown}}}
\newcommand{\go}[1]{{\underline{#1}}}
\newcommand{\com}{{\mathsf{com}}}
\newcommand{\weak}{{\mathsf{weak}}}
\newcommand{\cellst}{{\mathsf{cell\_ST}}}
\newcommand{\cellinit}{{\mathsf{cell\_init}}}

\newlength{\arrow}
\newcommand*{\constantwidthxrightarrow}[1]{\xrightarrow{\mathmakebox[\arrow]{#1}}}

\usepackage{everypage}
\usepackage{environ}
\usepackage{pdflscape}
\newcounter{abspage}

\makeatletter
\newcommand{\newSFPage}[1]% #1 = \theabspage
  {\global\expandafter\let\csname SFPage@#1\endcsname\null}

\NewEnviron{SidewaysFigure}{\begin{figure}[p]
\protected@write\@auxout{\let\theabspage=\relax}% delays expansion until shipout
  {\string\newSFPage{\theabspage}}%
\ifdim\textwidth=\textheight
  \rotatebox{90}{\parbox[c][\textwidth][c]{\linewidth}{\BODY}}%
\else
  \rotatebox{90}{\parbox[c][\textwidth][c]{\textheight}{\BODY}}%
\fi
\end{figure}}

\AddEverypageHook{% check if sideways figure on this page
  \ifdim\textwidth=\textheight
    \stepcounter{abspage}% already in landscape
  \else
    \@ifundefined{SFPage@\theabspage}{}{\global\pdfpageattr{/Rotate 0}}%
    \stepcounter{abspage}%
    \@ifundefined{SFPage@\theabspage}{}{\global\pdfpageattr{/Rotate 90}}%
  \fi}
\makeatother

\begin{document}

\maketitle

\begin{abstract}
The non-commutative sequoid operator $\sequoid$ on games was introduced  to capture algebraically the presence of state in history-sensitive strategies in game semantics, by imposing a causality relation on the  tensor product of games.  Coalgebras for the  functor  $A\sequoid\blank$  --- i.e., morphisms from $S$ to $A \sequoid S$ --- may be viewed as state transformers:  if $A\sequoid\blank$ has a \emph{final coalgebra}, $\oc A$, then the anamorphism of such a state transformer encapsulates its explicit state, so that it is shared only between successive invocations.

We study the  conditions under which a final coalgebra $\oc A$ for  $A\sequoid\blank$ is the carrier of a \emph{cofree commutative comonoid} on $A$. That is, it is a model of the exponential of linear logic in which we can construct imperative objects such as reference cells  coalgebraically, in a game semantics setting. We show that if the tensor \emph{decomposes} into the sequoid,  the final coalgebra $!A$ may be endowed with the structure of the cofree commutative comonoid if the natural isomorphism $\oc (A \times B) \cong \oc A \tensor \oc B$ holds. This condition is always satisfied if $\oc A$ is the \emph{bifree algebra} for $A\sequoid\blank$, but in general it is necessary to impose it, as we establish by giving an example of a sequoidally decomposable category of games in which plays will be allowed to have transfinite length. In this category, the final coalgebra for the functor $A\sequoid\blank$ is not the cofree commutative comonoid over $A$: we illustrate this by explicitly contrasting the final sequence for the functor $A\sequoid\blank$ with the chain of symmetric tensor powers used in the construction of the cofree commutative comonoid as a limit by  Melli\`es, Tabareau and Tasson. 

\end{abstract}

\section{Introduction}
Game semantics has been used to define a variety of models of higher-order programming languages with  mutable state, including Idealized Algol  \cite{SamsonGuyIAPassive}, and various fragments of ML \cite{AMV,AHM}. Unlike traditional denotational semantics, which typically represent imperative programs as state transformers, the state in these models is completely implicit: local declaration of mutable variables is interpreted as composition with a ``history sensitive'' strategy representing a reference cell. 
This is conceptually simple in principle but leads to some quite combinatorial definitions; a more explicit representation of the current state can be very useful for constructing and reasoning about imperative objects.

\subsection{Defining Higher-order Stateful Objects, Coalgebraically}
Let us first motivate the study of the coalgebraically derived cofree comonoid in game semantics by considering a similar but simpler and more familiar phenomenon. 
A \emph{state-transformer} in a symmetric monoidal category is a morphism $f:A \otimes S \rightarrow B \otimes S$ taking an argument together with an input state to a result together with an output state. A well-studied \cite{jacobs} technique in semantics is to  use an appropriate final coalgebra to \emph{encapsulate} the state in such a transformer, allowing multiple  successive invocations, each of which passes its output state as an input state to the next invocation. 

For example, consider the category $\Rel$ of sets and relations, with symmetric monoidal structure given by the Cartesian product (with unit $I$, the singleton set $\{*\}$). This has finite (bi)products (disjoint unions) so we may define the functor $F(A,S) = (A \otimes S) \oplus I$. For any object (set) $A$, let $A^*$ be the set of finite sequences of elements of $A$ (i.e. the carrier of the free monoid on $A$), and $\alpha:A^* \rightarrow F(A,A^*)$ be the morphism $\{(\varepsilon,\inr(*)\}\cup  \{(aw,(\inl(a,w))\ |\ a \in A,w \in A^*\}$. It is straightforward to show that: 
\begin{lemma}$(A^*,\alpha_A)$ is the final coalgebra  for $F(A,\_)$.
\end{lemma}
Since we have a natural transformation  $\inl_{A,S}:A \otimes S \rightarrow F(A,S)$, we may encapsulate the state in the  state transformer $f:S \rightarrow A \otimes S$ by taking the \emph{anamorphism} of $\tilde{f} = (f;\inl_{A,S})\cup\{(s,\inr (*))\suchthat s\in S\}\from S \rightarrow F(A,S)$, --- i.e. the unique $F(A,\_)$-coalgebra morphism from $(S,\tilde{f})$ into $(A^*,\alpha_{A})$. This  is  a morphism from an initial state $S$ into $A^*$: by definition, composing it with $\alpha:A \rightarrow F(A,A^*)$ (which we can think of as  \emph{invoking} our stateful object) returns a copy of $f$ together with the encapsulated morphism with updated internal state.   

 Distributivity of $\oplus$ over $\otimes$ implies that $F(A \oplus A',S) \cong F(A,S) \oplus F(A',S)$.  This allows state transformers to be aggregated, to construct stateful objects compounded of a series of methods which share access to a common state. 
For example, we may represent a reference cell storing integer values as a state transformer $\cell:\Na \rightarrow (\Na \oplus \Na) \otimes \Na$, obtained by aggregating two ``methods'' which share access to a value in $\Na$ representing the contents of the cell --- returning a ``read'' of the input state (and leaving it unchanged) or accepting a ``write'' of a new value and using it to update the state.   Thus (with appropriate tagging) it is the relation $\{(i, (\read(i),i))\suchthat i \in \Na\} \cup \{(i,({\mathtt{write}}(j),j))\suchthat i,j\in\Na\}$.
The anamorphism of the coalgebra $\cell^{\textasciitilde}\from\Na \rightarrow F(\Na \oplus \Na,\Na) $  is the relation from $\Na$ to $(\Na \oplus \Na)^*$ consisting of pairs of the form $(i_1,\read(i_1)^*{\mathtt{write}}(i_2)\read(i_2)^*\ldots)$. Composition with this morphism is precisely the interpretation of new variable declaration in the semantics  in $\Rel$ of the prototypical functional-imperative language \emph{Syntactic Control of Interference} (SCI) given in \cite{Mcsci}.

Coalgebraic methods thus give us a recipe for constructing and using categorical definitions of stateful semantic objects.  In order to fully exploit these, however, we endow $A^*$ with the structure of a \emph{comonoid} in our symmetric monoidal category, by defining morphisms $\delta_A:A^* \rightarrow A^* \otimes A^* = \{(u\cdot v,(u,v))\ |\ u,v \in A^*\}$ and $\epsilon:A \rightarrow I = \{(\varepsilon,*)\}$. In fact, this is the \emph{cofree comonoid} on $A$ --- there is a morphism $\eta_A:A^* \rightarrow A = \{(a,a)\ | \ a \in A\}$ such that for any comonoid $(B,\delta_B,\epsilon_B)$, composition with $\eta_A$ defines an equivalence (natural in $B$)  between the morphisms from $B$ into $A$, and the comonoid morphisms from $(B,\delta_B,\epsilon_B)$ into  $(A^*,\delta_A,\epsilon_A)$.   
\begin{proposition}$(A^*,\delta,\epsilon)$ is the cofree comonoid on $\Rel$.\footnote{The definitions of $\delta$ and $\epsilon$, and the proof that this is the cofree comonoid  may be derived from the fact that $(A^*,\alpha_A)$ is a \emph{bifree algebra} for $F(\_, A)$  --- i.e. $(A^*,\alpha^{-1})$ is an initial algebra for $F(A,\_)$ ($\alpha$ must be an isomorphism by Lambek's lemma). We leave this as an exercise.} 
\end{proposition}
This structure can be used to interpret procedures which share access to a stateful resource such as a reference cell.      
Its main limitation is that we have not defined a \emph{commutative} comonoid  for any non-empty set $A$ (evidently, $\delta$ is not invariant under post-composition with the symmetry isomorphism of the tensor). Thus we can only model procedures with shared access to the same stateful object if the order in which they are permitted to access it is fixed. (This is precisely the situation in SCI, where the typing system allows sharing  across sequential composition, but not between functions and their arguments.) In order to model sharing of state without this constraint (and build a Cartesian closed category), we need to endow our final coalgebra with the structure of a \emph{cofree commutative comonoid}, proposed as the basis of a model of linear logic by Lafont \cite{LafontCofCommCom}.  The category of sets and relations  does not allow this (the cofree commutative comonoid on an object $A$ in $\Rel$ is given by the set of finite multisets of $A$, which is not a final coalgebra). Hence, we turn to the richer structures of game semantics.

 \subsection{The cofree commutative comonoid as a final coalgebra}
We now outline the remainder of the paper. Our main contribution is an investigation of the circumstances in which the  cofree commutative comonoid on $A$ arises from a final coalgebra for the functor $A\sequoid\blank$, where $\sequoid$ (the sequoid) is a non-commutative operation on games  introduced by one of the authors \cite{laird02}.\footnote{We will focus on a particular category of ``history sensitive'', Abramsky-Jagadeesan style games \cite{abramskyjagadeesangames}, but sequoidal structure is a unifying feature of sequential, history-sensitive games: see \cite{laird02} for a variant of the Hyland-Ong games and \cite{Clairambault08aremark} for Conway games.} 
 In this setting, we  can model a state transformer for a program as a morphism $S \to A \sequoid S$ --- i.e., a coalgebra for the functor $A\sequoid\blank$. The final coalgebra for this functor is the exponential game $!A$ introduced by Hyland \cite{hyland1997games}, which corresponds to a $\omega$-fold sequence $A \sequoid (A \sequoid (A \sequoid \ldots ))$; under appropriate conditions, it is the carrier for the  cofree  commutative comonoid on $A$. We aim to characterize these conditions using just the categorical structure,  in order to capture a general class of models and to derive formal principles for coinductively proving  program equivalences. In a nutshell, we require that a certain natural morphism $\oc A \tensor \oc B \to \oc (A\times B)$ is an isomorphism. This can be used to show that $\oc\blank$ gives rise to a strong monoidal functor. Perhaps more surprisingly, this is sufficient to show that $\oc A$ is the cofree commutative comonoid.

This \emph{strong monoidal hypothesis} holds whenever $\oc A$ is a bifree algebra for $A\sequoid\blank$. But we are also interested in cases where $\oc A$ is not bifree --- for example, in categories of ``win games'' and \emph{winning strategies}  \cite{hyland1997games}, which lack the partial maps which can be shown to arise in the bifree case. 
To show that the strong monoidal hypothesis is necessary in general, we introduce a  sequoidal category of games with \emph{transfinite} plays in which it does not hold: because a transfinite interleaving of two sequences of length $\omega$ may have length greater than $\omega$, the final coalgebra for $A\sequoid\blank$  (corresponding  to only $\omega$-many copies of the game $A$) cannot be the carrier for the cofree commutative comonoid.

We compare the  coalgebraic construction of the cofree exponential to the explicit characterization of the latter given by Melli\`es, Tabareau and Tasson \cite{MelliesCofCommCom} as the limit of a chain of symmetric tensor powers. This chain exists in any decomposable sequoidal category: where its limit exists and is preserved by the tensor (the conditions required in \cite{MelliesCofCommCom}) it must be the final coalgebra for $A \sequoid\blank$. However, in our categories of transfinite games, and win games and winning strategies (which may be viewed as games of length $\omega+1$), the construction fails --- this limit is not the cofree commutative comonoid.  

\section{Sequoidal categories}

\subsection{Game semantics and the sequoidal operator}
We shall present a form of game semantics in the style of \cite{hyland1997games} and \cite{abramskyjagadeesangames}.  A game $A$ is given by a set $M_A$ of \emph{$P$-moves} and \emph{$O$-moves} and by a non-empty prefix-closed set $P_A\subset M_A^*$ of \emph{positions}, which are alternating sequences of $O$-moves and $P$-moves.  We shall adopt the rule that all positions must start with an $O$-move.  We call a position a \emph{$P$-position} if it ends with a $P$-move or is empty and an \emph{$O$-position} if it ends with an $O$-move. 

A \emph{strategy} for a game $A$ is a non-empty prefix-closed subset $\sigma$ of $P_A$ that is closed under $O$-replies to $P$-positions and which satisfies \emph{determinism}: if $sa,sb\in\sigma$, where $s$ is an $O$-position, then $a=b$.  

We build connectives on games as in \cite{abramskyjagadeesangames}.  The set of moves for a compound game is given by the disjoint union of the sets of moves for the individual sub-games, and the positions for each game are defined as follows:

\begin{description}
  \item[Product] If $(A_i\suchthat i\in I)$ is a collection of games, then we write $\prod_{i\in I}A_i$ for the game in which player $O$, on his first move, may play in any of the games $A_i$.  From then on, play continues in $A_i$.  If $A_1,A_2$ are games, we write $A_1\times A_2$ for $\prod_{i=1}^2 A_i$.  
  \item[Tensor Product] If $A,B$ are games,the tensor product $A\tensor B$ is played by playing the games $A$ and $B$ in parallel, where player $O$ may elect to switch games whenever it is his turn and continue play in the game he has switched to.
  \item[Linear implication] The implication $A\implies B$ is played by playing the game $B$ in parallel with the \emph{negation} of $A$ - that is, the game formed by switching the roles of players $P$ and $O$ in $A$.  Since play in the negation of $A$ starts with a $P$-move, player $O$ is forced to make his first move in the game $B$.  Thereafter, player $P$ may switch games whenever it is her turn.
\end{description}

It is well known (see \cite{abramskyjagadeesangames}, for example) that we may compose strategies $\sigma$ for $A\implies B$ and $\tau$ for $B\implies C$ to get a morphism $\comp\tau\sigma$ for $A\implies C$ and that this structure gives rise to a monoidal closed category where objects are games, morphisms from $A$ to $B$ are strategies for $A\implies B$ and the tensor product and linear implication are given by $A\tensor B$ and $A\implies B$.  We call this category $\G$.  $\G$ has all products, given by $\prod_{i\in I}A_i$ as above.  

The one non-standard connective we will use is the \emph{sequoid} connective from \cite{laird02}:

\begin{description}
    \item[Sequoid] If $A$ and $B$ are games, then the positions of $A\sequoid B$ are precisely the positions of $A\tensor B$ that are empty or that start with a move in $A$.
\end{description}

By inspection, we can verify that we have structural isomorphisms:
\[
  \begin{array}{cc}
    \dist\from A\tensor B\toisom (A\sequoid B)\times(B\sequoid A)
      & \dist^0 \from I \sequoid C \toisom I \\
    \dec\from(A\times B)\sequoid C\toisom (A\sequoid C)\times (B\sequoid C)
      & \run \from A \sequoid I \toisom A \\
    \passoc\from (A\sequoid B)\sequoid C\toisom A\sequoid (B\tensor C) &
  \end{array}
\]

We might expect that the sequoid would give rise to a functor $\G\times\G\to \G$ in the way that the tensor product does, through playing strategies in parallel.  However, this does not quite work: playing strategies $\sigma$ for $A\implies B$ and $\tau$ for $C \implies D$ in parallel does not necessarily give rise to a valid strategy for $(A\sequoid C)\implies (B \sequoid D)$, since player $P$ might end up playing in $C$ before anyone has played in $A$.  However, if we require that the strategy $\sigma$ is \emph{strict} --- that is, that player $P$'s reply (if any) to the opening move in $B$ is always a move in $A$ --- then we do get a valid strategy $\sigma\sequoid\tau$ for $(A\sequoid C)\implies (B\sequoid D)$ and, moreover, $\sigma\sequoid\tau$ is strict.  We shall write $\G_s$ for the category of games with \emph{strict} strategies as morphisms; then $\blank\sequoid\blank$ gives us a functor $\G_s\times \G\to\G_s$.  

\subsection{Sequoidal categories}

We now formalize these observations into a category-theoretic definition.  The purpose of this definition is to formalize precisely what it is about categories of games that makes them suitable for modeling stateful programs, and to give an equational characterization of the combinatorial definitions used in the Abramsky-McCusker model of Idealized Algol \cite{laird02, SamsonGuyIAPassive}.  

\begin{definition}
  A \emph{sequoidal category} consists of the following data:
  \begin{itemize}
    \item A symmetric monoidal category $\C$ with monoidal product $\tensor$ and tensor unit $I$, associators $\assoc_{A,B,C}\from(A\tensor B)\tensor C\toisom A\tensor(B\tensor C)$, unitors $\runit_A\from A\tensor I\toisom A$ and $\lunit_A\from I\tensor A\toisom A$ and braiding $\sym_{A,B}\from A\tensor B\to B\tensor A$.
    \item A category $\C_s$.
    \item A right monoidal category action \cite{Actegory} of $\C$ on the category $\C_s$.  That is, a functor $\blank\sequoid\blank\from\C_s\times\C\to\C_s$ that gives rise to a monoidal functor from $\C$ into the category of endofunctors on $\C_s$.  We write $\passoc_{A,B,C}\from (A \sequoid B) \sequoid C \to A \sequoid (B\tensor C)$ and $\run_A\from A \sequoid I \to A$ for the coherence parts of this monoidal functor.

    \item A functor $J\from \C_s\to\C$ (in the games example, this is the inclusion functor $\G_s\to\G$)

    \item A natural transformation $\wk_{A,B}\from J(A)\tensor B\to J(A\sequoid B)$ satisfying the coherence conditions\footnote{These coherence conditions say that $(J,\wk)$ is a \emph{lax morphism of right monoidal actions of $\C$} from the sequoidal action $(\C_s, \blank\sequoid\blank)$ to the `right multiplication' action $(\C, \blank\tensor\blank)$.}:
    \begin{equation*}
      \begin{tikzcd}
        J(A) \tensor I \arrow[r, "\runit_A" yshift=0.3em] \arrow[d, "\wk_{A,I}"']
          & J(A) \\
        J(A \sequoid I) \arrow[ur, "J(\run_A)"']
          &
      \end{tikzcd}
      \,\,\,
      \begin{tikzcd}
        (J(A) \tensor B) \tensor C \arrow[r, "\wk_{A,B}\tensor\id_C" yshift=0.3em] \arrow[d, "\assoc_{A,B,C}"']
          & J(A \sequoid B) \tensor C \arrow[r, "\wk_{A\sequoid B, C}" yshift=0.3em]
            & J((A \sequoid B) \sequoid C)\\
        J(A) \tensor (B \tensor C) \arrow[r, "\wk_{A,B\tensor C}"']
          & J(A \sequoid (B \tensor C)) \arrow[ur, "J(\passoc_{A,B,C})"']
            &
      \end{tikzcd}
    \end{equation*}
  \end{itemize}
\end{definition}

Our category of games satisfies further conditions:

\begin{definition}
  Let $\C=(\C,\C_s,J,\wk)$ be a sequoidal category.  We say that $\C$ is an \emph{inclusive sequoidal category} if $\C_s$ is a full-on-objects subcategory of $\C$ containing all isomorphisms and finite products of $\C$, and the morphisms $\wk_{A,B}$ and $J$ is the inclusion functor.

  We say that $\C$ is \emph{decomposable} if $I$ is a terminal object for $\C$ and if for any $A$ and $B$, the tensor product $A \tensor B$ is a Cartesian product of $A \sequoid B$ and $B \sequoid A$, with projections $\wk_{A,B}:A \tensor B \rightarrow A \sequoid B$ and  $\comp{\wk_{A,B}}{\sym_{A,B}}:A \tensor B \rightarrow B \sequoid A$.  We say that $\C$ is \emph{distributive} if whenever the product $\prod_{i\in I} A_i$ exists, then $\left(\prod_{i\in I}A_i\right)\sequoid B$ is the product of the $A_i\sequoid B$, with projections $\pr_i\sequoid\id_B$, and if it has a terminal object $1$ satisfying $1\sequoid A\cong 1$ for all objects $A$.  

Although there are important examples where we do not have products, our examples will all be categories with all products.  Then we can state the definitions of decomposability and distributivity more succinctly by requiring that the natural transformations
  \begin{gather*}
    \dec_{A,B} = \langle \wk_{A,B}, \comp{\wk_{A,B}}{\sym_{A,B}}\rangle\from A\tensor B\to (A\sequoid B)\times (B\sequoid A) \\
    \dec^0 \from I \to 1 \\
    \dist_{A,B,C} = \langle \pr_1\sequoid \id_C,\pr_2\sequoid\id_C\rangle\from (A\times B)\sequoid C\to (A\sequoid C)\times (B\sequoid C) \\
    \dist_{(A_i\colon i\in I),B} = \langle \pr_i\sequoid \id_C \suchthat i\in I\rangle \from \left(\prod_{i\in I}A_i\right)\sequoid C \to \prod_{i\in I}(A_i\sequoid C) \\
    \dist_{A,0}\from 1\sequoid A\to 1
  \end{gather*}
  are isomorphisms.
\end{definition}

 The category of games $\G$ and categories arising in different traditions of game semantics \cite{Clairambault08aremark,laird02} are the prototype examples of distributive, decomposable sequoidal categories.  

 \begin{remark}
     Churchill, Laird and McCusker give a result in \cite{martinsthesis} that implies that any sequoidal category satisfying these, and other, extra conditions can be used to model a proof calculus for describing games and strategies.  Nevertheless, they note that examples do exist of sequoidal categories that do not arise as categories of games, such as the category of locally Boolean domains described in \cite{LairdLbd} (see the last paragraph of that paper, and also the citation in \cite{martinsthesis}).  
 \end{remark}

\subsection{The sequoidal exponential}

There are several ways to add exponentials to the basic category of games, but the definition that fits our purposes is the one based on countably many copies of the base game (see \cite{hyland1997games}, for example): the exponential $\oc A$ of $A$ is the game in which player $O$ may switch between countably many copies of $A$ -- $A_0, A_1, A_2, \dots$, as long as he starts them in order, starting with $A_0$, then opening $A_1$ and so on.  This condition on the order in which games may be opened is very important, as it allows us to define the exponential morphisms $\oc A \to \oc A \tensor \oc A$ and $\oc A \to \oc \oc A$.  In the first case, it can be proved \cite{LairdCofCommCom} that the comultiplication $\oc A \to \oc A \tensor \oc A$ exhibits $\oc A$ as the \emph{cofree commutative comonoid} on $A$, which shows that $A$ is a suitable model for the exponential \cite{LafontCofCommCom}.  

Even more interestingly from the point of view of modelling stateful languages, we may characterize $\oc A$ as the \emph{final coalgebra} for the functor $J(A\sequoid\blank)\from \G\to \G$ (henceforth we shall write this functor as $A\sequoid\blank$, eliding the inclusion functor $J$).  That is, given a coalgebra from $\sequoid\blank$ --- a game $B$ and a morphism $\sigma\from B \to A \sequoid B$ --- we get a unique morphism $\fcoal{\sigma}$ making the following diagram commute:
\[
  \begin{tikzcd}
    B \arrow[r, "\sigma"] \arrow[d, "\fcoal{\sigma}"']
      & A \sequoid B \arrow[d, "\id_A\sequoid\fcoal{\sigma}"] \\
    \oc A \arrow[r, "\alpha"']
      & A \sequoid \oc A
  \end{tikzcd}
  \]
We call $\fcoal{\sigma}$ the \emph{anamorphism} of $\sigma$.  

We shall use the following standard pieces of coalgebra theory:
\begin{description}
  \item[Lambek's Lemma] $\alpha_A$ is an isomorphism, with inverse given by the anamorphism of the map $\id\sequoid\alpha_A\from A\sequoid\oc A\to A \sequoid (A \sequoid \oc A)$ \cite{Lambek}.  In particular, $\alpha_A$ is a morphism in $\G_s$.  In the general case, we deduce that $\alpha_A$ is a morphism in $\C_s$.  
  \item[Final Sequence] If $\C$ is a category with enough limits and $F$ is an endofunctor on $\C$, we may build up an ordinal indexed sequence of objects and morphisms of $\C$ (that is, a functor $\oppcat{\Ord}\to\C$, where $\Ord$ is the category of ordinals and prefix inclusions):
    \[
      1 \leftarrow F(1) \leftarrow F^2(1) \leftarrow \cdots F^\omega(1) \leftarrow F^{\omega+1}(1) \leftarrow \cdots
      \]
  (by repeatedly applying $F$ and taking limits).  If this sequence stabilizes for any $\delta$ (i.e., if the morphism from $F^{\delta+1}(1)\to F^\delta(1)$ is an isomorphism), then $F^\delta(1)$ is the final coalgebra for $F$ \cite{finalseq}.  In the case $F=A\sequoid\blank$, we shall write $A^{\sequoid\delta}$ for $F^\delta(1)$.  
\end{description}

\subsection{Imperative programs as anamorphisms}\label{Ipaa}

We now illustrate the construction of stateful objects using anamorphisms by constructing the strategy $\cell$ from \cite{SamsonGuyIAPassive} that represents a storage cell.  If $X$ is a set of values, write $\go{X}$ for the game denoting the corresponding type: that is, $X$ is the game with maximal plays $qx$, where $x$ ranges over the elements of $X$, and $\go{X}$ has canonical strategies $\go{x}$ for $x\in X$, in which player $P$ responds to the opening move $q$ with the move $x$.  In particular, the game $\go{\{*\}}$ corresponding to a singleton set denotes the void or command type $\com$. We shall write $\Sigma = \go{\{*\}}$ for this game, and write $\OK = \go{*}$ for its unique total strategy.  

Following \cite{SamsonGuyIAPassive}, we define $\Var[X]$ to be the type $\com^X\times X$; that is, the product of $X$-many copies of the command type with one copy of the type $X$.  We can think of this with an object that has a method $\wwrite\_x$ for each element $x$, together with a method $\read$.  The corresponding game is the game
\[
    \Var[X]=\Sigma^X \times \go{X}
  \]
In order to tell apart the various games, we shall write $\Sigma_x$ for each copy of $\Sigma$, and write the moves of $\Sigma_x$ as $q_x$ and $*_x$.  Let $d\in X$ be a fixed default value.  Then it is quite easy to describe what the strategy $\cell$ on $\oc\Var[X]$ should be: it is the strategy that always responds to $q_x$ with $*_x$ (as it is forced to do) and which responds to the move $q$ in $\go X$ with that value $x\in X$ such that $q_x$ has been played most recently (or with $d$ if player $O$ has not yet played in $\Sigma^X$).  This is more or less how the strategy is defined in \cite{SamsonGuyIAPassive}.  The problem is that the state (the current most recently written value of $x$) is \emph{implicit}, and it is hard to get a handle on it.  

Instead, we try a state-transformer based approach.  We shall use $\oc \go{X}$ to represent the state of the storage cell (the $\oc$ is there since we will need to refer to the state multiple times).  We define morphisms $\read\from \oc \go X\to \go X\sequoid \oc \go X$ and $\wwrite\_x\from \oc\go X \to \Sigma\times\oc\go X$ as follows: $\read$ is the canonical morphism $\alpha_{\go X}$, while $\wwrite\_x$ is the following composite, which throws away the previous state and updates it with the value $x$:
\[
  \oc X \xrightarrow{\weak} I \xrightarrow{\runit} I \tensor I \xrightarrow{\OK \tensor \oc \go{x}} \Sigma \tensor \oc\go{X} \xrightarrow{\wk} \Sigma\sequoid \oc\go{X}
  \]
Taking the product of these morphisms and applying the distributivity of $\times$ over $\sequoid$, we get our state transformer:
\[
  \cellst \from \oc \go X \xrightarrow{\langle \wwrite_x\colon x\in X\;,\; \read\rangle} (\Sigma\sequoid\oc \go X)^X \times \go X \sequoid \oc \go X \xrightarrow{\dist\inv} (\Sigma^X \times \go X) \sequoid \oc \go X
  \]
By the definition of $\oc\Var[X]$, there is now a unique morphism $\cellinit\from \oc X \to \oc\Var[X]$ making the following diagram commute:
\[
  \begin{tikzcd}
    \oc \go X \arrow[r, "\cellst"] \arrow[d, "\cellinit"']
      & \Var[X] \sequoid \oc \go X \arrow[d, "\id\sequoid\cellinit"] \\
    \oc \Var[X] \arrow[r, "\alpha"]
      & \Var[X] \sequoid \oc \Var[X]
  \end{tikzcd}
  \]
Define a strategy $\sigma\from \oc\go X\to \oc\Var[X]$ combinatorially by saying that $\sigma$ is the strategy that behaves like $\cell$ (as defined above) on $\oc\Var[X]$ but which interrogates its argument in order to establish the default value, rather than using a fixed value.  By inspection, we can verify that replacing $\cellinit$ with $\sigma$ in the diagram above makes the square commute, and therefore that $\cellinit=\sigma$ by uniqueness.  It follows that $\cell$ is equal to the following composite:
\[
  I=\oc I \xrightarrow{\oc \go{d}} \oc X \xrightarrow{\cellinit} \oc \Var[X]
  \]
But now we are able to reason about its state explicitly by using the coalgebraic definition.  

\begin{exercise}
    By modifying the definition of $\wwrite\_x$, show that the same construction may be used to model a stack with $\mathsf{push}$ and $\mathsf{pop}$ methods.
\end{exercise}

\section{Constructing cofree commutative comonoids in sequoidal categories}

\subsection{A formula for the sequoidal exponential}

We observed that the exponential $\oc A$ of a game $A$ arises as the final coalgebra for the functor $A\sequoid\blank$.  We also observed that $\oc A$ has the structure of a cofree commutative comonoid on $A$.  These two facts are both crucial if we want to use sequoidal categories to model stateful programs.  In this section, we shall consider conditions under which we may deduce that the final coalgebra for $A\sequoid\blank$ and the cofree commutative comonoid over $A$ coincide.

One important result is the formula  given by Melli\`es, Tabareau and Tasson \cite{MelliesCofCommCom}, which does not depend on the presence of Cartesian products but which obtains the cofree commutative comonoid as a limit of \emph{symmetric tensor powers}.
\begin{definition}If $A$ is an object in a symmetric monoidal category, a $n$-fold symmetric tensor power of $A$ is an  \emph{equalizer} $(A^n,\eq)$ for the group $G$ of symmetry automorphisms on $A^{\otimes n}$. A tensor power is preserved by the tensor product if $(B\otimes A^n,\id_B \otimes \eq)$ is an equalizer for the automorphisms $\{\id_B \otimes g\ |\ g \in G\}$.
\end{definition}
In any affine category\footnote{This is a special case of the situation considered in  \cite{MelliesCofCommCom}: that $A$ is a ``free pointed object''.} with symmetrized tensor powers of $A$ we may define a diagram $\Delta(A) = $
\[
  I \xleftarrow{p_0} A \xleftarrow{p_1} A^2 \xleftarrow{p_2}\cdots\xleftarrow{p_{i-1}}A^i\xleftarrow{p_i} \cdots
  \]
where   
$p_i:A^{i+1} \rightarrow A^i$ is the unique morphism given by the universal property of the symmetric tensor power,  such that $\comp{\eq_i:A^{i+1}}{p_i} \rightarrow A^{\otimes i} = \eq_{i+1};(A^{\otimes i} \otimes t_A)$. 

Melli\`es, Tabareau and Tasson \cite{MelliesCofCommCom} have shown that where the limit $(A^\infty, \{p^\infty_i\from A^\infty \rightarrow A^i\})$ for this diagram exists  and commutes with the tensor,  --- i.e. for each object $B$, $B \tensor A^\infty$ is the limit of 
\[
  B \tensor I \xleftarrow{\id_B\tensor p_0} B\tensor A \xleftarrow{\id_B\tensor p_1} B \tensor A^2 \xleftarrow{\id_B\tensor p_2} \cdots
  \]
then a comultiplication $\mu:A^\infty \rightarrow A^\infty \tensor A^\infty$ may be defined making $(A^\infty,\mu,t_{\oc A})$ the cofree commutative comonoid. Where these conditions are satisfied, we shall call this a MTT-exponential.   

In the category of games, the morphisms $\id\sequoid\pr_1\from A\sequoid(B\times C)\to A\sequoid B$ and $\id\sequoid\pr_2\from A\sequoid(B\times C)\to A\sequoid C$ are jointly monomorphic, and this joint monomorphism is preserved by the tensor product.  If a distributive sequoidal category satisfies the same property, we say that it is \emph{strong distributive}.
\begin{proposition}Any strong distributive decomposable sequoidal category has all symmetric tensor powers, and these are preserved by the tensor.  
\end{proposition} 
\begin{proof}
By sequoidal decomposability, for any $n\in\mathbb N$, $A^{\tensor (n+1)}$ is the Cartesian product $\prod_{i \le n} (\id_A \sequoid A^{\tensor n})$ with projections 
$\sym_i;\wk_{A,A^{\tensor n}} $, where  $\sym_i: A^{\tensor (n+1)} \rightarrow A^{\tensor (n+1)}$ is the symmetry isomorphism corresponding to the permutation on $n$ which swaps $1$ and $i$. 

  Define  $\wk^n: A^{\tensor n}\rightarrow A^{\sequoid n}$ by $\wk^{n+1} = \wk_{A, A^{\tensor n}}; (\id_A \sequoid \wk^n)$.  We show (by induction on $n$) that for any $n$, the morphisms $\comp{\wk_n}{\sym^\pi}$ are jointly monomorphic, where $\sym^\pi$ ranges over all of the permutation isomorphisms  on $A^{\tensor n}$, and that this joint monomorphism is preserved by the tensor product.

  We define the equalizer $\eq_n:A^{\sequoid n} \rightarrow A^{\tensor n}$ inductively by setting $\eq_n$ to be the product $\langle \id_A\sequoid\eq_{n-1}, \dots, \id_A\sequoid\eq_{n-1}\rangle$, using the identification of $A^{\sequoid n}$ as a product given above.  We may show inductively that $\comp{\wk^n}{\comp{\sym^\pi}{\eq_n}}=\id$ for all permutations $\pi\in S_n$.  
  Given any $f:C \rightarrow A^{\tensor n} \tensor B$ such that $\comp{(\sym^\pi \tensor \id_B)}{f} = f$ for any permutation $\pi$, taking $\comp{(\wk^n \tensor \id_B)}{f}:C \rightarrow A^{\sequoid n} \tensor B$ gives the unique morphism such that $\comp{(\eq_n\tensor \id_B)}{\comp{(\wk^n\tensor \id_B)}{f}} = f$.  Indeed, for all permutations $\pi\in S_n$ we have $\comp{((\comp{\comp{\wk_n}{\sym^\pi}}{\comp{\eq_n}{\wk^n}})\tensor\id_B)}{f} = \comp{(\wk^n\tensor\id_B)}{f}=\comp{((\comp{\wk^n}{\sym^\pi})\tensor\id_B)}{f}$.  Hence, $\comp{(\eq_n\tensor \id_B)}{\comp{(\wk^n\tensor \id_B)}{f}} = f$.  For uniqueness, use the fact that $\eq_n$ is a left inverse for $\wk^n$. 
\end{proof}
Thus, in any strong distributive sequoidally decomposable category, the diagram $\Delta(A)$ exists for any $A$. If a limit $A^\infty$ for this diagram exists and is preserved by the tensor 
--- i.e. for any $B$, $B \tensor A^\infty$ is the limit for $B \tensor \Delta(A)$ --- then it is the cofree commutative comonoid.  

Moreover, by distributivity, preservation by the tensor product implies preservation by the sequoid, which tells us that in this case the final sequence for $A\sequoid\blank$ must converge at $\omega$ and that the limit $A^\infty$ must be the carrier for the final coalgebra for $A\sequoid\blank$.  

\subsection{Win-games and winning strategies}

The construction from \cite{MelliesCofCommCom} covers a lot of important cases, but there are some situations in which the right conditions are not satisfied. Instead, we shall need the techniques that we shall describe in the following sections.  One example in which we cannot use the Melli\`es-Tabareau-Tasson formula is that of \emph{win-games}, or games with a winning condition \cite{abramskyjagadeesangames,martinsthesis}.  Given a game $A$, we write $\overline{P_A}$ to be the limit-closure of $P_A$ --- that is, $P_A$ together with the set of infinite sequences, all of whose finite prefixes are in $P_A$.  A \emph{win-game} is a game $A$ together with a function $\zeta_A\from \overline{P_A}\to\OP$ such that:
\begin{itemize}
  \item $\zeta_A(\emptyplay) = P$
  \item $\zeta_A(sa) = \lambda_A(a)$
\end{itemize}
Thus, $\zeta_A$ is entirely determined on $P_A$, and the only new information is the values that $\zeta_A$ takes on the infinite positions in $\overline{P_A}$.  The reason we bother to define $\zeta_A$ on finite positions at all is so that we can define it easily on the connectives:
\[
  \begin{array}{cccccc}
    \multicolumn{3}{c}{
    \zeta_{A\tensor B}(s) = \zeta_A(s\vert_A) \wedge \zeta_B(s\vert_B)
    }
      & \multicolumn{3}{c}{
        \zeta_{A \implies B}(s) = \zeta_A(s\vert_A) \Rightarrow \zeta_B(s\vert_B)
      }\\[8pt]
    \multicolumn{2}{c}{
    \zeta_{\prod_{i\in I}A_i}(s) = \bigwedge_{i\in I}\zeta_{A_i}(s\vert_{A_i})
    }
      & \multicolumn{2}{c}{
        \zeta_{A\sequoid B}(s) = \zeta_A(s\vert_A) \wedge \zeta_B(s\vert_B)
      }
        & \multicolumn{2}{c}{
        \zeta_{\oc A}(s) = \bigwedge_{i\in I}(\zeta_A(s\vert_i))
        }
  \end{array}
\]

Here, $\wedge$ and $\Rightarrow$ are the usual propositional connectives on $\{T,F\}$, where we identify $T$ with $P$ and $F$ with $O$.  The infinite positions $s$ with $\zeta_A(s)=P$ are the \emph{$P$-winning} positions, while the infinite positions $s$ with $\zeta_A(s)=O$ are the \emph{$O$-winning} positions.  

We define a \emph{winning strategy} on $(A,\zeta_A)$ to be a total strategy $\sigma$ on $A$ such that every infinite sequence arising as the limit of sequences in $\sigma$ is a $P$-winning position.  It is known (see \cite{abramskyjagadeesangames}) that the composition of winning strategies is winning and that we get a decomposable, distributive sequoidal category $\W$ with $\oc A$ as the final coalgebra for $A\sequoid\blank$ and the cofree commutative comonoid over $A$ \cite{martinsthesis}.

However, in this case, $\oc A$ is not the sequential limit of the symmetrized tensor powers over $A$.  Since $\W$ is a decomposable, strong distributive sequoidal category, the symmetrized tensor powers of $A$ are given by the sequoidal powers $A^{\sequoid n}$.  But now the limit of these objects is not quite the game $\oc A$; instead, it is the game $A^{\sequoid\omega}=\co A$ in which player $O$ may open an arbitrarily large finite number of copies of $A$, but loses if he opens infinitely many.  In the finite case, there was no way to keep track of infinite positions, so we could not make this distinction, but in the win-games case we can: we set $\zeta_{\co A}(s)=\zeta_{\oc A}(s)$, unless $s$ contains moves in infinitely many games, in which case we set $\zeta_{\co A}(s) = P$.  

This limit is not preserved by the functor $A\sequoid\blank$: in the game $A^{\sequoid(\omega+1)}=A\sequoid\co A$, player $O$ wins if he wins either in $A$ or in $\co A$, so he can win even if he plays in infinitely many games, as long as he wins in the first copy of $A$.  Similarly, in the game $A^{\sequoid(\omega+n)}$, player $O$ wins as long as he wins in one of the first $n$ copies of $A$ or opens finitely many copies.  Therefore, the limit $A^{\sequoid\omega2}$ \emph{is} the game $\oc A$: the final sequence for $A\sequoid\blank$ in $\W$ stabilizes at $\omega2$ and, consequently, the exponential in $\W$ is not an MTT-exponential.  

This example is a special case of our later result on transfinite games.  For now, we shall examine a coalgebraic approach that will prove that the final coalgebra $\oc A$ for $A\sequoid \blank$ in the category $\W$ of win-games gives us a cofree commutative comonoid.  

\subsection{The coalgebraic construction under the strong monoidal hypothesis}

We shall now need to assume that we are in a decomposable, distributive sequoidal category $(\C,\C_s,J,\wk)$ such that $\C_s$ has all products and $J$ preserves them.  However, we shall no longer need the MTT assumption that the exponential should be constructed as a limit of sequoidal powers.  The main cost is that we shall need to make a further assumption: that a certain naturally defined morphism $\oc A \tensor \oc B\to \oc (A\times B)$ is an isomorphism.  This assumption, broadly corresponding to the demand that the functor $\oc A$ be strong monoidal from the Cartesian category $(\C,\times,1)$ to the monoidal category $(\C,\tensor,I)$, will allow us to construct the comultiplication directly from the Cartesian structure and the definition of $\oc A$ as a final coalgebra.  

\begin{notation}
  We shall sometimes make the monoidal structure of the Cartesian product explicit by writing $\sigma\times\tau$ for $\langle\comp\sigma{\pr_1},\comp\tau{\pr_2}\rangle$.
\end{notation}

\begin{definition}
  Let $A,B$ be objects of an decomposable, distributive sequoidal category $(\C,\C_s,J,\wk)$ with final coalgebras $\oc A\xrightarrow{\alpha_A} A\sequoid\oc A$ for all endofunctors of the form $A\sequoid\blank$.  Let $A,B$ be objects of $C$.  Then we have a composite $\kappa_{A,B}$:
  \settowidth{\arrow}{\scriptsize$\id_{A\sequoid(\oc A\tensor\oc B)}\times (\id_B\sequoid\sym_{\oc B,\oc A})$}
  \begin{IEEEeqnarray*}{RCL}
    \kappa_{A,B} = \oc A \tensor \oc B & \constantwidthxrightarrow{\dec_{A,B}} & (\oc A \sequoid \oc B) \times (\oc B \sequoid \oc A) \\
    \cdots & \constantwidthxrightarrow{(\alpha_A\sequoid\id_{\oc B}) \times (\alpha_B\sequoid\id_{\oc A})} & ((A \sequoid \oc A) \sequoid \oc B) \times ((B \sequoid \oc B) \sequoid \oc A) \\
    \cdots & \constantwidthxrightarrow{\passoc_{A,\oc A,\oc B}\inv\times\passoc_{B,\oc B,\oc A}\inv} & (A \sequoid (\oc A \tensor \oc B)) \times (B \sequoid (\oc B \tensor \oc A)) \\
      \cdots & \constantwidthxrightarrow{\id_{A\sequoid(\oc A\tensor\oc B)}\times (\id_B\sequoid\sym_{\oc B,\oc A})} & (A \sequoid (\oc A \tensor \oc B)) \times (B \sequoid (\oc A \tensor \oc B))
  \end{IEEEeqnarray*}
  We get a morphism $\comp{\dist\inv}{\kappa_{A,B}}\from \oc A \tensor \oc B \to (A\times B) \sequoid(\oc A \tensor \oc B)$ and we write $\int_{A,B}=\fcoal{\comp{\dist\inv}{\kappa_{A,B}}}\from \oc A \tensor \oc B \to \oc (A\times B)$. 
\end{definition}

\begin{proposition}
  In the category of games, the morphism $\int_{A,B}$ is an isomorphism for all games $A,B$.
\end{proposition}
\begin{proof}
  Observe that the morphism $\int_{A,B}$ is the copycat strategy on $\oc A\tensor\oc B\implies \oc (A\times B)$ that starts a copy of $A$ on the left whenever a copy of $A$ is started on the right and starts a copy of $B$ on the left whenever a copy of $B$ is started on the right (indeed, the morphisms in the diagram above are all copycat morphisms, so the copycat strategy we have just described must make that diagram commute).  Since there are infinitely many copies of both $A$ and $B$ available in $\oc (A\times B)$, and since a new copy of $A$ or $B$ may be started at any time, we may define an inverse copycat strategy on $\oc(A\times B)\implies \oc A\tensor\oc B$.
\end{proof}
Our first main result for this section will be the following:
\begin{theorem}
  \label{StrongMonoidalFunctor}
  Let $(\C,\C_s,J,\wk)$ be a distributive and decomposable sequoidal category with a final coalgebra $\oc A\xrightarrow{\alpha_A}A\sequoid\oc A$ for each endofunctor of the form $A\sequoid\blank$.  Suppose further that the morphism $\int_{A,B}$ as defined above is an isomorphism for all objects $A,B$.  Then $A\mapsto\oc A$ gives rise to a strong symmetric monoidal functor from the monoidal category $(\C, \times, 1)$ to the monoidal category $(\C, \tensor, I)$.  
\end{theorem}

We start off by defining a morphism $\mu\from\oc A\to\oc A\tensor\oc A$.  This will turn out to be the comultiplication for the cofree commutative comonoid over $A$.  First, we note that we have the following composite:
\[
  \oc A \xrightarrow{\alpha_A} A \sequoid \oc A \xrightarrow{\Delta} (A\sequoid\oc A)\times (A\sequoid\oc A)\xrightarrow{\dist\inv} (A\times A)\sequoid\oc A
  \]
where $\Delta$ is the diagonal map on the product.  We set $\sigma_A=\fcoal{\comp{\dist\inv}{\comp{\Delta}{\alpha_A}}}\from\oc A \to \oc (A\times A)$ and we set $\mu_A=\comp{\int_{A,A}\inv}{\sigma_A}\from \oc A \to \oc A \tensor \oc A$.  

We also define a morphism $\der_A\from \oc A \to A$.  Note that since $I$ is isomorphic to $1$, we have a unique morphism $*_A\from A\to I$ for each $A$.  We define $\der_A$ to be the composite
\[
  \begin{tikzcd}
    \oc A \arrow[r, "\alpha_A"]
      & A \sequoid \oc A \arrow[r, "\id_A \sequoid *_{\oc A}" yshift=0.3em]
        & A \sequoid I \arrow[r, "\run_A"]
          & A
  \end{tikzcd}
  \]

We define the action of $\oc$ on morphisms as follows: suppose that $f\from A \to B$ is a morphism in $\C$.  Then we have a composite
\[
  \begin{tikzcd}
    \oc A \arrow[r, "\mu_A"]
      & \oc A \tensor \oc A \arrow[r, "\der_A\tensor\id_{\oc A}"]
        & A \tensor \oc A \arrow[r, "f\tensor\id_{\oc A}"]
          & B \tensor \oc A \arrow[r, "\wk_{B,\oc A}"]
            & B \sequoid \oc A
  \end{tikzcd}
  \]
and so we may define $\oc f$ to be the anamorphism $\fcoal{\comp{\wk_{B,\oc A}}{\comp{f\tensor\id_{\oc A}}{\comp{\der_A\tensor\id_{\oc A}}{\mu_A}}}}\from \oc A \to \oc B$.  

\begin{proposition}\label{StrongMonoidalFunctorActual}
  $f\mapsto\oc f$ respects composition, so $\oc$ is a functor.  Moreover, $\oc$ is a strong symmetric monoidal functor from the Cartesian category $(\C, \times, 1)$ to the symmetric monoidal category $(\C, \tensor, I)$, witnessed by $\int$ and $\epsilon$, where $\epsilon$ is the anamorphism of the composite $I \xrightarrow{runit} I \tensor I \xrightarrow{*\tensor\id} 1\tensor I\xrightarrow{\wk}1\sequoid I$ (this composite is an isomorphism, so $\epsilon$ is as well).
\end{proposition}
\begin{proof}
  See Appendix.
\end{proof}

This completes the proof of Theorem \ref{StrongMonoidalFunctor}.

Since $\oc$ is a strong monoidal functor, it induces a functor $\CCom(\oc)$ from the category $\CCom(\C,\times,1)$ of comonoids over $(\C, \times, 1)$ to the category $\CCom(\C,\tensor,I)$ of comonoids over $(\C, \tensor, I)$ making the following diagram commute:
\[
  \begin{tikzcd}
    \CCom(\C,\times,1) \arrow[r, "\F"] \arrow[d, "\CCom(\oc)"']
      & (\C, \times, 1) \arrow[d, "\oc"] \\
    \CCom(\C, \tensor, I) \arrow[r, "\F"]
      & (\C, \tensor, I)
  \end{tikzcd}
  \]
where $\F$ is the forgetful functor.

Let $A$ be an object of $\C$.  Since $(\C, \times, 1)$ is Cartesian, the diagonal map $\Delta\from A\to A\times A$ is the cofree commutative comonoid over $A$ in $(\C,\times,1)$.  

\begin{proposition}\label{itsMu}
  $\CCom(\oc)\left(A\xrightarrow{\Delta}A\times A\right)$ has comultiplication given by $\mu_A\from\oc A \to \oc A\tensor\oc A$ and counit given by the unique morphism $\eta_A\from \oc A \to I$.
\end{proposition}

In particular, this proves that the comultiplication $\mu_A$ is associative and that the counit $\eta_A$ is a valid counit for $\mu_A$.

We can now state our second main result from this section.

\begin{theorem}
  \label{Coalgebra__CoCoCo}
  Let $(\C,\C_s,J,\wk)$ be a sequoidal category satisfying all the conditions from Theorem \ref{StrongMonoidalFunctor}.  Let $A$ be an object of $\C$ (equivalently, of $\C_s$).  Then $\oc A$, together with the comultiplication $\mu_A$ and counit $\eta_A$, is the cofree commutative comonoid over $A$.
\end{theorem}

\subsection{The Sequoidal Exponential as a Bifree Algebra}
Observe that in our category of games, $(\oc A,\alpha)$ is in fact a \emph{bifree algebra} for $A\sequoid\blank$ --- the isomorphism $\alpha^{-1}: A \sequoid \oc A \rightarrow \oc A$ is an initial algebra for  $A \sequoid \blank$.  We may show that in such cases, the condition that $\oc $ is strong monoidal --- and thus the cofree exponential ---  always holds\footnote{Without requiring our sequoidally decomposable category to have finite products  we may equip each object $\oc A$ with the structure of a comonoid by defining:   
$\mu:\oc A \rightarrow \oc A  \otimes \oc A$
 to be the catamorphism of the $A \sequoid \blank$ algebra:
\[
  A \sequoid (\oc A \otimes \oc A) 
\to A \sequoid (\oc A \otimes \oc A) \times  A \sequoid (\oc A \otimes \oc A) \cong (A \sequoid \oc A) \sequoid \oc A \times  (A \sequoid \oc A) \sequoid \oc A \cong (\oc A \sequoid \oc A) \times (\oc A \sequoid \oc A) \cong (\oc A \otimes \oc A)
\]
This satisfies the further requirements of a \emph{linear category} in the sense of \cite{Bierman}, although it does not appear to be possible to show that it is the cofree commutative comonoid.} :
we may define an inverse to $\int:\oc A \otimes \oc B \rightarrow \oc  (A \times B)$  as the \emph{catamorphism} of the  $A \sequoid\blank$-algebra:
\[
  (\comp\dist{\kappa_{A,B}})\inv\from (A\times B) \sequoid (\oc A \tensor \oc B) \to \oc A \tensor \oc B
  \]

It is not necessary for the final $A \sequoid\_$-coalgebra to be bifree for the exponential to be strong monoidal and thus the cofree commutative comonoid. An example is provided by  the category $\W$ of win-games and winning strategies, which is \emph{sequoidal closed} (the restriction of the functor  $A \multimap\blank$ to strict strategies is right adjoint to $\blank \sequoid A$, and inclusion sends this adjunction to the usual monoidal closure). To show that the  final $A \sequoid\blank$-coalgebra in this category is not bifree, it suffices to observe that from such an algebra, we may derive a  \emph{fixed point} operator $\fix_A:\C(A,A) \rightarrow \C(I,A)$ for each $A$, such that $\comp{f}{\fix_A(f)} = \fix_A(f)$.      
\begin{proposition}Suppose $\C$ is sequoidal closed and decomposable, and $(\oc A,\alpha)$ is a bifree $A \sequoid\_$-algebra. Then we may define a fixed point operator on $A$.  
\end{proposition}    
\begin{proof}For any $A$, let  $\Phi_A:\oc (A \multimap A) \rightarrow A$ be the catamorphism of the counit to the adjunction $A \sequoid\blank \dashv A \multimap$,  $\epsilon_{A,A}: (A \multimap A) \sequoid A \rightarrow A$, which is a  $(A \multimap A) \sequoid \blank$-algebra. For any morphism $f:A \rightarrow A$ we may define $\fix_A(f) = \oc\Lambda(f);\Phi_A$, where $\Lambda(f):I \rightarrow (A \multimap A)$ is the ``name'' of $f$.    
\end{proof}
As one would expect, it is not possible to define a fixed point operator on the category of games and winning strategies ---  for example, if $\bot$ is the game with a single move then the hom-set $\C(I,\bot)$ is empty and hence there can be no morphism $\fix_\bot(\id_\bot)$. So the final  $A \sequoid \blank$-coalgebra is not bifree in this case. 

\section{Transfinite Games}

Of the conditions that we used to construct the cofree commutative comonoid in sequoidal categories, the requirement that $\int_{A,B}$ be an isomorphism stands out as the least satisfactory.  All the other conditions are `finitary', and relate directly to the connectives we have introduced, whereas the morphism $\int_{A,B}$ can only be constructed using the final coalgebra property for the exponential connective $\oc$.  For this reason, we might wonder whether we can do without the condition that $\int_{A,B}$ be an isomorphism.  In this section, we shall give a negative answer to that question: we shall construct a distributive and decomposable sequoidal closed category with final coalgebras $\oc A$ for all functors of the form $A\sequoid\blank$, and shall show that $\oc A$ does not have a natural comonoid structure.  In doing this, we hope to shed some light upon alternative algebraic or coalgebraic constructions for the cofree commutative comonoid that work in a purely `finitary' manner.

\begin{warning}
    Although we shall refer to the final coalgebra for $A\sequoid\blank$ as $\oc A$ to avoid introducing new notation, this object will not be the carrier for a linear exponential comonad (i.e., a model of the exponential from Linear Logic) in the category.  
\end{warning}

Our sequoidal category will be closely modelled upon the category of games we have just considered: the objects will be games, with the modification that sequences of moves may now have transfinite length.  This is a natural construction, occurring in the study of determinacy by Mycielski \cite{mycielskiChoice}, Blass \cite{blassChoice} and Weiss \cite{longgamesbook}.  Transfinite games were used by Berardi and de'Liguoro to give a characterization of total functionals \cite{BerardiTransfiniteGames}, and they appear to the authors to be present in the semantic context in the work of Roscoe \cite{RoscoeCspInfinite}, Levy \cite{LevyGsInfinite} and Laird \cite{LairdOrdinalGames}.  

The general idea is as follows: we will show that the definition of the final coalgebra for the sequoid functor in a category of transfinite games is largely unchanged from the definition in the category of games with finite-length plays: $\oc A$ is the game formed from a countably infinite number of copies of $A$, indexed by $\omega$, with the proviso that player $O$ must open them in order.  We observe that the copycat strategy $\int_{A,B}\from\oc A\tensor\oc B\to \oc(A\times B)$ is not an isomorphism, and that we cannot construct the comultiplication $\oc A \to \oc A \tensor \oc A$ in a sensible way.  Moreover, we cannot construct the comonad $\oc A \to \oc \oc A$, so $\oc$ does not give us a model of linear logic in even the most general sense.  In all three cases, the reason why the construction fails is that we might run out of copies of the game $A$ (or $B$) on the left hand side before we have run out of copies on the right hand side.  In the finite-plays setting, it is impossible to run out of copies of a subgame, because there are infinitely many copies, so it is impossible to play in all of them in a finite-length play.  In the transfinite setting, however, we cannot guarantee this: consider, for example, a position in $\oc A_0 \implies \oc A_1 \tensor \oc A_2$ (with indices given so we can refer to the different copies of $A$) in which player $O$ has opened all the copies of $A$ in $\oc A_1$.  Since player $P$ is playing by copycat, she must have opened all of the copies of $A$ in $\oc A_0$.  If, at time $\omega+1$, player $O$ now plays in $\oc A_2$, player $P$ will have no reply to him.

The `correct' definition of $\oc A$ in the transfinite game category is one in which there is an unlimited number of copies of $A$ to open (rather than $\omega$-many), but this is not the final coalgebra for the functor $A\sequoid\blank$.  

\subsection{Transfinite Games}

We give a brief summary of the construction of the category of transfinite games.  

We shall fix an additively indecomposable ordinal $\alpha = \omega^\beta$ throughout, which will be a bound on the ordinal length of positions in our game.  So, for example, the original category of games is the case $\alpha=\omega$.  If $X$ is a set, we write $X^{*<\alpha}$ for the set of transfinite sequences of elements of $X$ of length less than $\alpha$.

\begin{definition}
  A (completely negative) \emph{game} or a \emph{game over $\alpha$} or an \emph{$\alpha$-game} is given by a forest (i.e., a prefix-closed set) $P_A$ of alternating transfinite sequences of $O$-moves and $P$-moves of length less than $\alpha$ and a function $\zeta_A\from P_A\to\OP$ that designates each position as an $O$-position or a $P$-position.  We require that $sa$ is a $P$-position if $a$ is a $P$-move and an $O$-position if $a$ is an $O$-move, so $\zeta_A$ only gives us information about plays of limiting length.  
  
  $P_A$ is subject to a continuity condition: if $s$ is a sequence of moves whose length is a limit ordinal and $t\in P_A$ for all proper prefixes $t\pprefix s$, then $s\in P_A$.  

  We say that a game $A$ is \emph{completely negative} if every position of limiting length is a $P$-position.  
\end{definition}

\begin{definition}
  A \emph{strategy} for an $\alpha$-game $A$ is a non-empty prefix-closed subset of $P_A$ that satisfies closure under $O$-replies and the determinism condition, just as for finite strategies.  
\end{definition}

We can form the product, tensor product, sequoid, exponential and linear implication in the same way that we do for finite games.  The $\zeta$-functions are extended to connectives according to the propositional formulae given for win-games above.  If $A$ and $B$ are completely negative, then so are $A\times B$, $A\tensor B$, $A\sequoid B$ and $\oc A$, but $A\implies B$ might not be completely negative.  

Given games $A,B,C$, strategies $\sigma$ for $A\implies B$ and $\tau$ for $B\implies C$, we may compose $\sigma$ and $\tau$ in the same way that we compose strategies for finite games (but we have to use the fact that $\alpha$ is additively indecomposable so that we can ensure that the interleaving of sequences of length less than $\alpha$ still has length less than $\alpha$).  

We can show that this composition is associative and moreover that we obtain a distributive and decomposable sequoidal category whose objects are completely negative games.  We call this category $\G(\alpha)$ and call the corresponding strict subcategory $\G_s(\alpha)$.  The hardest part of this is showing that the category is monoidal closed, because the linear implication of completely negative games is not necessarily completely negative.  It turns out that there is always a `minimal' completely negative game extending $A\implies B$, which gives us the monoidal closed structure, but we will not discuss this here, since monoidal closedness is not particularly important to any of our constructions.

\subsection{The final sequence for the sequoidal exponential}

We now want to show that $\G(\alpha)$ has a final coalgebra for each functor $A\sequoid\blank$, given by the transfinite game $\oc A$, which is defined as follows:
\begin{itemize}
  \item $M_{\oc A} = M_A \times \omega$
  \item $\lambda_{\oc A} = \lambda_A \circ \pr_1$
\end{itemize}
We define $\oc P_A$ to be the set of all sequences $s\in M_{\oc A}^{*<\alpha}$ such that $s\vert_n\in P_A$ for all $n$ and such that every move in $A_{n+1}$ occurs later than some move in $A_n$.  Then we define $\zeta_{\oc A}\from\oc P_A\to\OP$ by
\[
  \zeta_{\oc A}(s) = \bigwedge_{n\in\omega}\zeta_A(s\vert_n)
  \]
In other words, $\zeta_{\oc A}(s)=P$ if and only if $\zeta_A(s\vert_n)=P$ for all $n$.  We define $P_{\oc A}$ to be the set of all sequences in $\oc P_A$ that are alternating with respect to $\zeta_{\oc A}$.  

There is a natural copycat strategy $\alpha_A\from \oc A \to A \sequoid \oc A$, just as in the finite plays case.  We want to show that this is the final coalgebra for $A\sequoid\blank$.  The proof for the finite case found in \cite{martinsthesis} will not work in this case, since it implicitly uses the fact that $\oc A$ is an MTT-exponential.  In the transfinite categories, this is no longer the case.  

While it is possible to prove that $\alpha_A\from \oc A\to A\sequoid \oc A$ is the final coalgebra for $A\sequoid\blank$ directly, we shall instead give a proof by extending the MTT sequence to the full final sequence.  We shall give a complete classification of the games $A^{\sequoid\gamma}$ and use it to show that the final sequence for $A\sequoid\blank$ must stabilize at $\oc A$.  

\begin{definition}
  Let $s\in\omega^{*<\alpha}$ be any transfinite sequence of natural numbers.  We define the \emph{derivative} $\Delta s$ of $s$ to be the sequence given by removing all instances of $0$ from $s$ and subtracting $1$ from all other terms.  In other words, if $s\from\gamma\to\omega$, for $\gamma<\alpha$, then we have:
  \[
    \Delta s = s\inv(\omega\setminus\{0\})\xrightarrow{s}\omega\setminus\{0\}\xrightarrow{-1}\omega
    \]
  (where $s\inv(\omega\setminus\{0\})$ carries the induced order).  We now define predicates $\blank\le\gamma$ on sequences $s\in \omega^{*<\alpha}$ as follows:
  \begin{itemize}
    \item $\emptyplay\le 0$
    \item If $\Delta s\le\gamma$, then $s\le\gamma+1$
    \item If $\mu$ is a limit ordinal and $s\in\omega^{*<\alpha}$ is such that for all successor-length prefixes $t\prefix s$ we have $t\le\gamma$ for some $\gamma<\mu$, then $s\le\mu$.  In other words, $\{s\in \omega^{*<\alpha}\suchthat s\le\mu\}$ is the limit-closure of the union of the sets $\{s\in\omega^{*<\alpha}\suchthat s\le\gamma\}$ for $\gamma<\mu$.
  \end{itemize}
\end{definition}

It is easy to prove some basic results about these predicates:
\begin{proposition}\label{basicResultsAboutRank}
  i) If $s\le\gamma$ and $t$ is any subsequence of $s$ (not necessarily an initial prefix), then $t\le\gamma$.
  
  ii) If $s\le\gamma$, then $\Delta s\le\gamma$

  iii) If $s\le\gamma$ and $\gamma\le\delta$, then $s\le\delta$

  iv) If $s\in\omega^{*<\alpha}$ has length $\mu$, where $\mu$ is a limit ordinal, then $s\le\mu$.  If $s$ has length $\mu+n$ for some $n\in\omega$, then $s\le\mu+\omega$.  In particular, $s\le\alpha$ for all $s\in \omega^{*<\alpha}$.  
\end{proposition}
\begin{proof}
  Left as an exercise.  
\end{proof}

We can then classify the terms of the final sequence for $A\sequoid\blank$ as follows:
\begin{theorem}\label{finalseqClassification}
  Let $A$ be any game.  Then $A^{\sequoid\gamma}\cong(M_{\oc A}, \lambda_{\oc A}, \zeta_{\oc A}, P_{\oc A,\gamma})$, where
  \[
    P_{\oc A,\gamma} = \{s\in P_{\oc A}\suchthat \pr_2\circ s\le\gamma\}
    \]
  The morphism $j_\gamma^\delta$ is the copycat strategy.
\end{theorem}

\begin{corollary}
  The final sequence for $A\sequoid\blank$ stabilizes at $\alpha$ and we have $A^{\sequoid\alpha}=\oc A$.
\end{corollary}
\begin{proof}
  By Proposition \ref{basicResultsAboutRank}(iv), $\pr_2\circ s\le\alpha$ for all $s\in P_{\oc A}$ and so $\pr_2\circ s\le(\alpha+1)$, by Proposition \ref{basicResultsAboutRank}(iii).  It follows, by Theorem \ref{finalseqClassification}, that $A^{\sequoid\alpha}=\oc A$ and that the morphism $A^{\sequoid\alpha}\to A^{\sequoid(\alpha+1)}$ is the morphism $\alpha_A$.
\end{proof}

In particular, $\oc A$ is the carrier of the final coalgebra for $A\sequoid\blank$.  But, as we saw before, it is not the carrier for the cofree commutative comonoid over $A$; indeed, more is true:

\begin{proposition}
  $\oc A$ does not carry the structure of a linear exponential comonad.
\end{proposition}
\begin{proof}
  Indeed, if it did, then \cite{SchalkWhatIs} we would have an isomorphism
  \[
    \oc (A\times B) \cong \oc A \tensor \oc B
    \]
  But this isomorphism does not hold for suitably long transfinite games.  For example, if $A$ and $B$ are bounded games (i.e., games such that the lengths of plays are bounded by some finite number $n$) containing at least one play of length $2$ then $\oc A\tensor \oc B$ contains plays of length $\omega2$ (play through all the copies of $A$, then through all the copies of $B$), while the lengths of plays in $\oc (A\times B)$ are bounded by $\omega$, since $A\times B$ is a bounded game.
\end{proof}

%%
%% Bibliography
%%

%% Either use bibtex (recommended), 

\bibliography{bibliography}\vfill

\begin{thebibliography}{10}

\bibitem{abramskyjagadeesangames}
Samson Abramsky and Radha Jagadeesan.
\newblock Games and full completeness for multiplicative linear logic.
\newblock {\em The Journal of Symbolic Logic}, 59(2):543--574, 1994.
\newblock URL: \url{http://arxiv.org/abs/1311.6057}.

\bibitem{SamsonGuyIAPassive}
Samson Abramsky and Guy McCusker.
\newblock Full abstraction for {Idealized Algol} with passive expressions.
\newblock {\em Theor. Comput. Sci.}, 227(1-2):3--42, September 1999.
\newblock URL: \url{http://dx.doi.org/10.1016/S0304-3975(99)00047-X}, \href
  {http://dx.doi.org/10.1016/S0304-3975(99)00047-X}
  {\path{doi:10.1016/S0304-3975(99)00047-X}}.

\bibitem{AMV}
S.~Abramsky{ and G. McCusker}.
\newblock Call-by-value games.
\newblock In M.~Neilsen and W.~Thomas, editors, {\em Proceedings of CSL '97},
  pages 1--17. Springer-{V}erlag, 1998.
\newblock \href {http://dx.doi.org/10.1007/BFb0028004}
  {\path{doi:10.1007/BFb0028004}}.

\bibitem{AHM}
S.~Abramsky{, K. Honda and G. McCusker}.
\newblock A fully abstract games semantics for general references.
\newblock In {\em Proceedings of LICS '98}. IEEE Press, 1998.
\newblock \href {http://dx.doi.org/10.1109/LICS.1998.705669}
  {\path{doi:10.1109/LICS.1998.705669}}.

\bibitem{BerardiTransfiniteGames}
Stefano Berardi and Ugo de'Liguoro.
\newblock Total functionals and well-founded strategies.
\newblock In Jean-Yves Girard, editor, {\em Typed Lambda Calculi and
  Applications: 4th International Conference, TLCA'99 L'Aquila, Italy, April
  7--9, 1999 Proceedings}, pages 54--68. Springer Berlin Heidelberg, Berlin,
  Heidelberg, 1999.
\newblock URL: \url{http://dx.doi.org/10.1007/3-540-48959-2_6}, \href
  {http://dx.doi.org/10.1007/3-540-48959-2_6}
  {\path{doi:10.1007/3-540-48959-2_6}}.

\bibitem{Bierman}
G.~M. Bierman.
\newblock {\em What is a categorical model of Intuitionistic Linear Logic?},
  pages 78--93.
\newblock Springer Berlin Heidelberg, Berlin, Heidelberg, 1995.
\newblock URL: \url{http://dx.doi.org/10.1007/BFb0014046}, \href
  {http://dx.doi.org/10.1007/BFb0014046} {\path{doi:10.1007/BFb0014046}}.

\bibitem{blassChoice}
Andreas Blass.
\newblock Equivalence of two strong forms of determinacy.
\newblock {\em Proceedings of the American Mathematical Society},
  52(1):373--376, 1975.
\newblock URL: \url{http://www.jstor.org/stable/2040166}.

\bibitem{martinsthesis}
Martin Churchill, Jim Laird, and Guy McCusker.
\newblock Imperative programs as proofs via game semantics.
\newblock {\em CoRR}, abs/1307.2004, 2013.
\newblock URL: \url{http://arxiv.org/abs/1307.2004}.

\bibitem{Clairambault08aremark}
Pierre Clairambault.
\newblock A remark on jim laird’s games category for general references,
  2008.

\bibitem{Hasuo07generictrace}
Ichiro Hasuo, Bart Jacobs, and Ana Sokolova.
\newblock Generic trace semantics via coinduction.
\newblock In {\em Logical Methods in Comp. Sci}, page 2007, 2007.

\bibitem{hyland1997games}
Martin Hyland.
\newblock Game semantics.
\newblock {\em Semantics and logics of computation}, 14:131, 1997.

\bibitem{longgamesbook}
Neeman Itay.
\newblock {\em The Determinacy of Long Games}.
\newblock De Gruyter CY, 2008.
\newblock URL: \url{//www.degruyter.com/view/product/178683}.

\bibitem{jacobs}
Bart Jacobs.
\newblock {\em Introduction to Coalgebra: Towards Mathematics of States and
  Observation}, volume~59 of {\em Cambridge Tracts in Theoretical Computer
  Science}.
\newblock Cambridge University Press, 2016.

\bibitem{Actegory}
G.~Janelidze and G.M. Kelly.
\newblock A note on actions of a monoidal category.
\newblock {\em Theory and Applications of Categories [electronic only]},
  9:61--91, 2001.
\newblock URL: \url{http://eudml.org/doc/122510}.

\bibitem{LafontCofCommCom}
Yves Lafont.
\newblock {\em Logiques, cat\'{e}gories et machines}.
\newblock PhD thesis, Universit\'{e} Paris 7, 1988.

\bibitem{laird02}
J.~Laird.
\newblock A categorical semantics of higher-order store.
\newblock In {\em Proceedings of CTCS '02}, number~69 in ENTCS. Elsevier, 2002.

\bibitem{LairdLbd}
J.~Laird.
\newblock Locally boolean domains.
\newblock {\em Theoretical Computer Science}, 342(1):132 -- 148, 2005.
\newblock URL:
  \url{http://www.sciencedirect.com/science/article/pii/S0304397505003427},
  \href {http://dx.doi.org/http://dx.doi.org/10.1016/j.tcs.2005.06.007}
  {\path{doi:http://dx.doi.org/10.1016/j.tcs.2005.06.007}}.

\bibitem{LairdOrdinalGames}
J.~Laird.
\newblock Sequential algorithms for unbounded nondeterminism.
\newblock {\em Electronic Notes in Theoretical Computer Science}, 319:271 --
  287, 2015.
\newblock URL:
  \url{http://www.sciencedirect.com/science/article/pii/S1571066115000845},
  \href {http://dx.doi.org/http://dx.doi.org/10.1016/j.entcs.2015.12.017}
  {\path{doi:http://dx.doi.org/10.1016/j.entcs.2015.12.017}}.

\bibitem{LairdCofCommCom}
James Laird.
\newblock Functional programs as coroutines: A semantic analysis.
\newblock {\em Logical methods in computer science}.
\newblock To appear.

\bibitem{jimNonAffine}
James Laird.
\newblock From qualitative to quantitative semantics by change of base.
\newblock In {\em 20th International Conference on Foundations of Software
  Science and Computation Structures}, April 2017.
\newblock URL: \url{http://opus.bath.ac.uk/53626/}.

\bibitem{Lambek}
Joachim Lambek.
\newblock A fixpoint theorem for complete categories.
\newblock {\em Mathematische Zeitschrift}, 103(2):151--161, 1968.
\newblock URL: \url{http://dx.doi.org/10.1007/BF01110627}, \href
  {http://dx.doi.org/10.1007/BF01110627} {\path{doi:10.1007/BF01110627}}.

\bibitem{LevyGsInfinite}
Paul~Blain Levy.
\newblock Infinite trace equivalence.
\newblock {\em Annals of Pure and Applied Logic}, 151(2):170 -- 198, 2008.
\newblock URL:
  \url{http://www.sciencedirect.com/science/article/pii/S0168007207000887},
  \href {http://dx.doi.org/http://dx.doi.org/10.1016/j.apal.2007.10.007}
  {\path{doi:http://dx.doi.org/10.1016/j.apal.2007.10.007}}.

\bibitem{Mcsci}
G.~McCusker.
\newblock A fully abstract relational model of {S}yntactic {C}ontrol of
  {I}nterference.
\newblock In {\em Proceedings of Computer Science Logic '02}, number 2471 in
  LNCS. Springer, 2002.

\bibitem{MelliesCofCommCom}
Paul-Andr{\'e} Melli{\`e}s, Nicolas Tabareau, and Christine Tasson.
\newblock {\em An Explicit Formula for the Free Exponential Modality of Linear
  Logic}, pages 247--260.
\newblock Springer Berlin Heidelberg, Berlin, Heidelberg, 2009.
\newblock URL: \url{http://dx.doi.org/10.1007/978-3-642-02930-1_21}, \href
  {http://dx.doi.org/10.1007/978-3-642-02930-1_21}
  {\path{doi:10.1007/978-3-642-02930-1_21}}.

\bibitem{mycielskiChoice}
Jan Mycielski.
\newblock On the axiom of determinateness.
\newblock {\em Fundamenta Mathematicae}, 53(2):205--224, 1964.
\newblock URL: \url{http://eudml.org/doc/213736}.

\bibitem{RoscoeCspInfinite}
A.~W. Roscoe.
\newblock Unbounded non-determinism in {CSP}.
\newblock {\em Journal of Logic and Computation}, 3(2):131, 1993.
\newblock URL: \url{http://dx.doi.org/10.1093/logcom/3.2.131}, \href
  {http://dx.doi.org/10.1093/logcom/3.2.131}
  {\path{doi:10.1093/logcom/3.2.131}}.

\bibitem{SchalkWhatIs}
Andrea Schalk.
\newblock What is a categorical model for linear logic?, October 2004.
\newblock URL: \url{http://www.cs.man.ac.uk/~schalk/notes/llmodel.pdf}.

\bibitem{finalseq}
James Worrell.
\newblock On the final sequence of a finitary set functor.
\newblock {\em Theoretical Computer Science}, 338(1):184 -- 199, 2005.
\newblock URL:
  \url{http://www.sciencedirect.com/science/article/pii/S0304397504008023},
  \href {http://dx.doi.org/http://dx.doi.org/10.1016/j.tcs.2004.12.009}
  {\path{doi:http://dx.doi.org/10.1016/j.tcs.2004.12.009}}.

\end{thebibliography}

%% .. or use the thebibliography environment explicitly

\pagebreak

\appendix

\section{Proofs}

\subsection{Proof of Proposition \ref{StrongMonoidalFunctorActual}}

\begingroup
\def\thetheorem{\ref{StrongMonoidalFunctorActual}}
\begin{proposition}
  $\sigma\mapsto\oc\sigma$ respects composition, so $\oc$ is a functor.  Moreover, $\oc$ is a strong symmetric monoidal functor from the Cartesian category $(\C, \times, 1)$ to the symmetric monoidal category $(\C, \tensor, I)$, witnessed by $\int$ and $\epsilon$, where $\epsilon$ is the anamorphism of the composite $I \xrightarrow{runit} I \tensor I \xrightarrow{*\tensor\id} 1\tensor I\xrightarrow{\wk}1\sequoid I$.
\end{proposition}
\addtocounter{theorem}{-1}
\endgroup

In order to show that $\sigma\mapsto\oc\sigma$ respects composition, we need the following lemma:

\begin{lemma}\label{aFormulaForAlpha}
  Let $A$ be an object of $\C$.  Then $\alpha_A\from\oc A\to A\sequoid\oc A$ is equal to the following composite:
  \[
    \begin{tikzcd}
      \oc A \arrow[r, "\mu_A"]
        & \oc A \tensor \oc A \arrow[r, "\der_A\tensor\id_{\oc A}"]
          & A \tensor \oc A \arrow[r, "\wk_{A,\oc A}"]
            & A \sequoid \oc A
    \end{tikzcd}
    \]
\end{lemma}

\begin{proof}
  We may paste together the anamorphism diagrams for $\int_A$ and $\sigma_A$ to form the following diagram (where we shall omit subscripts where there is no ambiguity):
  \begin{equation*}
    \begin{tikzcd}[cramped]
      \oc A \arrow[r, "\alpha"] \arrow[d, "\sigma_A"']
        & A \sequoid \oc A \arrow[r, "\Delta"]
          &[-24pt] (A \sequoid \oc A) \times (A \sequoid \oc A) \arrow[r, "\dist\inv"]
            & (A \times A) \sequoid \oc A \arrow[d, "\id_{A\times A}\sequoid\sigma_A"] \\
      \oc(A\times A) \arrow[rrr, "\alpha"]
        &
          &
            & (A \times A) \sequoid \oc(A\times A) \\
      \oc A \tensor \oc A \arrow[rr, "\kappa_{A,A}"] \arrow[u, "\int_A"]
        &
          & (A \sequoid (\oc A \tensor \oc A)) \times (A \sequoid (\oc A \tensor \oc A)) \arrow[r, "\dist\inv"]
            & (A\times A)\sequoid(\oc A \tensor \oc A) \arrow[u, "\id_{A\times A}\sequoid\int_A"']
    \end{tikzcd}
  \end{equation*}
  where we observe that the composites down the left and right hand sides (after inverting the lower arrows) are $\mu_A$ and $\id_{A\times A}\sequoid\mu_A$.

  Now note that we have the following commutative square:
  \[
    \begin{tikzcd}
      (A \times A) \sequoid \oc A \arrow[r, "\dist"] \arrow[d, "\id_{A\times A}\sequoid\mu_A"']
        & (A \sequoid \oc A) \times (A \sequoid \oc A) \arrow[d, "(\id\sequoid\mu)\times(\id\sequoid\mu)"] \\
      (A \times A) \sequoid (\oc A \tensor \oc A) \arrow[r, "\dist"']
        & (A \sequoid (\oc A \tensor \oc A)) \times (A \sequoid (\oc A \tensor \oc A))
    \end{tikzcd}
    \]

  (using the definition of $\dist$).  Putting this together with the diagram above, we get the following commutative diagram:
  \[
    \begin{tikzcd}
      \oc A \arrow[r, "\alpha"] \arrow[d, "\mu_A"']
        & A \sequoid \oc A \arrow[r, "\Delta"]
          & (A \sequoid \oc A) \times (A \sequoid \oc A) \arrow[d, "\id\sequoid\mu_A\times\id\sequoid\mu_A"] \\
      \oc A \tensor \oc A \arrow[rr, "\kappa_{A,A}"]
        & \cdots \arrow[r]
          & (A \sequoid (\oc A \tensor \oc A)) \times (A \sequoid (\oc A \tensor \oc A))
    \end{tikzcd}
    \]
  We now expand the definition of $\kappa_{A,A}$  and take the projections on to the first and second components, yielding the following two commutative diagrams:
  \[
    \begin{tikzcd}\label{firstProjectionDiagram}
      \oc A \arrow[rrr, "\alpha"] \arrow[d, "\mu_A"']
        &
          &
            & A \sequoid \oc A \arrow[d, "\id\sequoid\mu_A"] \\
      \oc A \tensor \oc A \arrow[r, "\alpha\tensor\id"]
        & (A \sequoid \oc A) \tensor \oc A \arrow[r, "\wk"]
          & (A \sequoid \oc A) \sequoid \oc A \arrow[r, "\passoc\inv"]
            & A \sequoid (\oc A \tensor \oc A)
    \end{tikzcd}\tag{\textbf{1}}
    \]

  \begin{equation*}\label{secondProjectionDiagram}
    \begin{tikzcd}[cramped]
     \oc A \arrow[rrrr, "\alpha"] \arrow[d, "\mu_A"']
        &
          &
            &
              & A \sequoid \oc A \arrow[d, "\id\sequoid\mu_A"] \\
      \oc A \tensor \oc A \arrow[r, "\sym"']
        & \oc A \tensor \oc A \arrow[r, "\alpha\tensor\id" yshift=0.4em]
          & (A \sequoid \oc A) \tensor \oc A \arrow[r, "\comp{\passoc\inv}{\wk}" yshift=0.4em]
            & A \sequoid (\oc A \tensor \oc A) \arrow[r, "\id\sequoid\sym"']
              & A \sequoid (\oc A \tensor \oc A)
    \end{tikzcd}\tag{\textbf{2}}
  \end{equation*}

  From diagram \eqref{firstProjectionDiagram}, we construct the following commutative diagram:
  
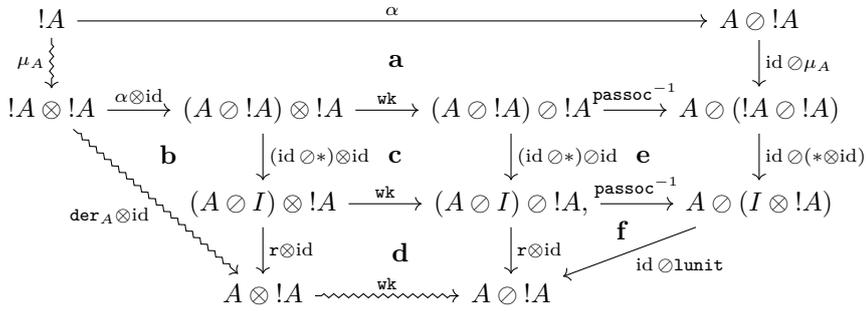
\[
    \begin{tikzcd}
      \oc A \arrow[rrr, "\alpha"] \arrow[d, rightsquigarrow, "\mu_A"'] \arrow[phantom, drrr, "\textbf{a}" xshift=0.9em]
        &
          &
            & A \sequoid \oc A \arrow[d, "\id\sequoid\mu_A"] \\
      \oc A \tensor \oc A \arrow[r, "\alpha\tensor\id"] \arrow[ddr, rightsquigarrow, "\der_A\tensor\id"'] \arrow[phantom, dr, "\textbf{b}" xshift=0.6em]
        & (A \sequoid \oc A) \tensor \oc A \arrow[r, "\wk"] \arrow[d, "(\id\sequoid*)\tensor\id"] \arrow[phantom, dr, "\textbf{c}" xshift=0.3em]
          & (A \sequoid \oc A) \sequoid \oc A \arrow[r, "\passoc\inv"] \arrow[d, "(\id\sequoid*)\sequoid\id"] \arrow[phantom, dr, "\textbf{e}" xshift=0.3em]
            & A \sequoid (\oc A \sequoid \oc A) \arrow[d, "\id\sequoid(*\tensor\id)"] \\
        & (A \sequoid I) \tensor \oc A \arrow[r, "\wk"] \arrow[d, "\run\tensor\id"] \arrow[phantom, dr, "\textbf{d}" xshift=0.3em]
          & (A \sequoid I) \sequoid \oc A \arrow[r, "\passoc\inv"], \arrow[d, "\run\tensor\id"] \arrow[phantom, r, "\textbf{f}" xshift=-0.6em, yshift=-1.0em]
            & A \sequoid (I \tensor \oc A) \arrow[dl, "\id\sequoid\lunit"] \\
        & A \tensor \oc A \arrow[r, rightsquigarrow, "\wk"]
          & A \sequoid \oc A
            &
    \end{tikzcd}
    \]

  \textbf{a} is diagram \eqref{firstProjectionDiagram}.

  \textbf{b} commutes by the definition of $\der_A$.

  \textbf{c} and \textbf{d} commute because $\wk$ is a natural transformation.

  \textbf{e} commutes because $\passoc$ is a natural transformation.

  \textbf{f} commutes by one of the coherence conditions in the definition of a sequoidal category.

  We now observe that the composite of the three squiggly arrows is the composite we are trying to show is equal to $\alpha$; we have $\alpha$ along the top, so it will suffice to show that the composite
  \[
    \begin{tikzcd}
      \xi_A\;=\;\oc A \arrow[r, "\mu_A"]
        & \oc A \tensor \oc A \arrow[r, "*\tensor\id"]
          & I \tensor \oc A \arrow[r, "\lunit"]
            & \oc A
    \end{tikzcd}
    \]
  is equal to the identity.  We do this using diagram \eqref{secondProjectionDiagram}.  First we construct the diagram shown in Figure \ref{hugeDiagram1}.

  \begin{SidewaysFigure}
    \[
      \begin{tikzcd}[ampersand replacement=\&]
        \oc A \arrow[rrrrr, "\alpha"] \arrow[d, "\mu_A"'] \arrow[phantom, drrrrr, "\textbf{a}"]
          \&
            \&
              \&
                \& 
                  \& A \sequoid \oc A \arrow[d, "\id\sequoid\mu_A"] \\
        \oc A \tensor \oc A  \arrow[r, "\sym"] \arrow[d, "*\tensor\id"'] \arrow[phantom, dr, "\textbf{b}"]
          \& \oc A \tensor \oc A \arrow[r, "\alpha\tensor\id"] \arrow[d, "\id\tensor*"]
            \& (A \sequoid \oc A) \tensor \oc A \arrow[r, "\wk"] \arrow[d, "\id\sequoid*"] \arrow[phantom, dr, "\textbf{d}"]
              \& (A \sequoid \oc A) \sequoid \oc A \arrow[r, "\passoc\inv"] \arrow[d, "\id\sequoid*"] \arrow[phantom, dr, "\textbf{e}"]
                \& A \sequoid (\oc A \tensor \oc A) \arrow[r, "\id\sequoid\sym"] \arrow[d, "\id\sequoid(\id\sequoid*)"] \arrow[phantom, dr, "\textbf{c}"]
                  \& A \sequoid (\oc A \tensor \oc A) \arrow[d, "\id\sequoid (*\sequoid\id)"] \\
        I \tensor \oc A \arrow[r, "\sym"] \arrow[d, "\lunit"'] \arrow[phantom, dr, "\textbf{g}"]
          \& \oc A \tensor I \arrow[r, "\alpha\tensor\id"] \arrow[d, "\runit"] \arrow[phantom, dr, "\textbf{f}"]
            \& (A \sequoid \oc A) \tensor I \arrow[r, "\wk"] \arrow[d, "\runit"] \arrow[phantom, dr, "\textbf{i}"]
              \& (A \sequoid \oc A) \sequoid I \arrow[r, "\passoc\inv"] \arrow[d, "\run"] \arrow[phantom, dr, "\textbf{j}"]
                \& A \sequoid (\oc A \tensor I) \arrow[r, "\id\sequoid\sym"] \arrow[d, "\id\sequoid\runit"] \arrow[phantom, dr, "\textbf{h}"]
                  \& A \sequoid (I \tensor \oc A) \arrow[d, "\id\sequoid\lunit"] \\
        \oc A \arrow[r, "\id"']
          \& \oc A \arrow[r, "\alpha"]
            \& A \sequoid \oc A \arrow[r, "\id"']
              \& A \sequoid \oc A \arrow[r, "\id"']
                \& A \sequoid \oc A \arrow[r, "\id"']
                  \& A \sequoid \oc A
      \end{tikzcd}
      \]
    \caption{
      \textbf{a} is diagram \eqref{secondProjectionDiagram}. \\[\baselineskip]
      \textbf{b} and \textbf{c} commute because $\sym$ is a natural transformation, \textbf{d} commutes because $\wk$ is a natural transformation and \textbf{e} commutes because $\passoc$ is a natural transformation.  \textbf{f} commutes because $\runit$ is a natural transformation. \\[\baselineskip]
      \textbf{g} and \textbf{h} commute by one of the coherence conditions for a symmetric monoidal category.  \textbf{i} commutes by one of the coherence conditions for $\wk$ in the definition of a sequoidal category and \textbf{j} commutes by one of the coherence conditions for $\passoc$ in the definition of a sequoidal category.
    }\label{hugeDiagram1}
  \end{SidewaysFigure}

  Now observe that the composite $\xi_A$ is running along the left hand side of Figure \ref{hugeDiagram1}, while $\id\sequoid\xi$ is running along the right.  Since we have $\alpha$ along the bottom, it follows by the uniqueness of $\fcoal{\cdot}$ that $\xi=\fcoal{\alpha}=\id_{\oc A}$.  
\end{proof}

Now we are ready to show that $f\mapsto\oc f$ respects composition.  Let $A,B,C$ be objects, let $f$ be a morphism from $A$ to $B$ and let $g$ be a morphism from $B$ to $C$.   Using Lemma \ref{aFormulaForAlpha} and the definition of $\oc f$, $\oc g$, we may construct a commutative diagram:
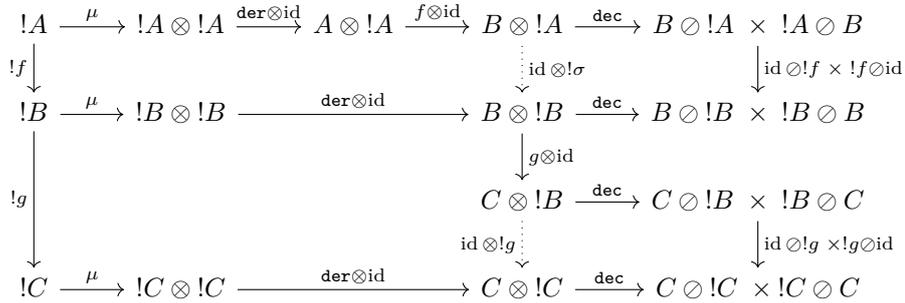
\[
  \begin{tikzcd}
    \oc A \arrow[r, "\mu"] \arrow[d, "\oc f"']
      & \oc A \tensor \oc A \arrow[r, "\der\tensor\id"]
        & A \tensor \oc A \arrow[r, "f\tensor\id"]
          & B \tensor \oc A \arrow[r, "\dec"] \arrow[d, dotted, "\id\tensor\oc\sigma"]
            & B \sequoid \oc A \;\times\; \oc A \sequoid B \arrow[d, "\id\sequoid\oc f\;\times\;\oc f\sequoid\id"] \\
    \oc B \arrow[r, "\mu"] \arrow[dd, "\oc g"']
      & \oc B \tensor \oc B \arrow[rr, "\der\tensor\id"]
        &
          & B \tensor \oc B \arrow[r, "\dec"] \arrow[d, "g\tensor\id"]
            & B \sequoid \oc B \;\times\; \oc B \sequoid B \\
      &
        &
          & C \tensor \oc B \arrow[r, "\dec"] \arrow[d, dotted, "\id\tensor\oc g"']
            & C \sequoid \oc B \;\times\; \oc B \sequoid C \arrow[d, "\id\sequoid\oc g\;\times\oc g\sequoid\id"] \\
    \oc C \arrow[r, "\mu"]
      & \oc C \tensor \oc C \arrow[rr, "\der\tensor\id"]
        &
          & C \tensor \oc C \arrow[r, "\dec"]
            & C \sequoid \oc C \;\times \oc C \sequoid C
  \end{tikzcd}
  \]
Here, the outermost (solid) shapes are the product of shapes that commute by the definition of $\oc f$, $\oc g$ (after we have replaced $\alpha_B$, $\alpha_C$ with the composite from Lemma \ref{aFormulaForAlpha}).  The smaller squares on the right hand side commute because $\dec$ is a natural transformation.  Since $\dec$ is an isomorphism, the two rectangles on the left commute as well.

Throwing away the right hand squares and adding some new arrows at the right, we arrive at the following commutative diagram:
\[
  \begin{tikzcd}
    \oc A \arrow[r, "\comp{(f \tensor\id)}{\comp{(\der\tensor\id)}\mu}"] \arrow[d, "\oc f"']
      &[72pt] B \tensor \oc A \arrow[r, "g \tensor\id"] \arrow[d, "\id\tensor\oc f"']
        & C \tensor \oc A \arrow[d, "\id\tensor\oc f"] \\
    \oc B \arrow[d, "\oc g"']
      & B \tensor\oc B \arrow[r, "g\tensor\id"] \arrow[d, "g\tensor\oc\tau"']
        & C \tensor \oc B \arrow[dl, "\id\tensor\oc g"] \\
    \oc C \arrow[r, "\comp{(\der\tensor\id)}{\mu}"]
      & C \tensor \oc C
        &
  \end{tikzcd}
  \]
We have just shown that the square on the left commutes.  The shapes on the right commute by inspection.  We now throw away the internal arrows and apply $\wk$ on the right hand side:
\[
  \begin{tikzcd}
    \oc A \arrow[r, "\comp{((\comp g f)\tensor\id)}{\comp{(\der\tensor\id)}\mu}"] \arrow[d, "\oc f"']
      &[72pt] C \tensor \oc A \arrow[d, "\id\tensor\oc f"'] \arrow[r, "\wk"]
        & C \sequoid \oc A \arrow[d, "\id\sequoid\oc f"] \\
    \oc B \arrow[d, "\oc g"']
      & C \tensor \oc B \arrow[d, "\id\tensor\oc g"'] \arrow[r, "\wk"]
        & \oc C \sequoid \oc B \arrow[d, "\id\sequoid\oc g"] \\
    \oc C \arrow[r, "\comp{(\der\tensor\id)}{\mu}"]
      & C \tensor \oc C \arrow[r, "\wk"]
        & C \sequoid\oc C
  \end{tikzcd}
  \]
By Lemma \ref{aFormulaForAlpha}, the composite along the bottom is equal to $\alpha_C$.  Therefore, by uniqueness of $\fcoal{\cdot}$, we have
\[
  \comp{\oc g}{\oc f} = \fcoal{\comp{\wk}{\comp{((\comp g f)\tensor\id)}{\comp{(\der\tensor\id)}\mu}}} = \oc(\comp g f)
  \]
Therefore, $\oc$ is indeed a functor.

We now want to show that $\oc$ has the structure of a strong symmetric monoidal functor from $(\C,\times,1)$ to $(\C,\tensor,I)$.  The relevant morphisms are:
\[
  \int_{A,B}\from \oc A \tensor \oc B \to \oc (A \times B)
  \quad
  \epsilon\from I \to \oc 1
  \]
By hypothesis, these are both isomorphisms.  We just need to show that the appropriate coherence diagrams commute.  That is, for any games $A,B,C$, we need to show that the following diagrams commute:
\[
  \begin{tikzcd}
    (\oc A \tensor \oc B) \tensor \oc C \arrow[r, "\assoc_{A,B,C}"] \arrow[d, "\int_{A,B}\tensor\id_{\oc C}"']
      &[24pt] \oc A \tensor (\oc B \tensor \oc C) \arrow[d, "\id_{\oc A}\tensor\int_{B,C}"] \\
    \oc (A \times B) \tensor \oc C \arrow[d, "\int_{A\times B,C}"']
      & \oc A \tensor \oc(B \times C) \arrow[d, "\int_{A,B\times C}"] \\
    \oc ((A \times B) \times C) \arrow[r, "\oc\assoc_{\times,A,B,C}"']
      & \oc (A \times (B \times C))
  \end{tikzcd}
\]
\[
  \begin{tikzcd}
    I \tensor \oc A \arrow[r, "\varepsilon\tensor\id"]
      & \oc 1 \tensor \oc A \arrow[d, "\int_{1,A}"] \\
    \oc A \arrow[u, "\lunit_{\oc A}"] \arrow[r, "\oc\lunit_A"]
      & \oc (1 \times A)
  \end{tikzcd}\quad
  \begin{tikzcd}
    \oc A \tensor I \arrow[r, "\id\tensor\varepsilon"]
      & \oc A \tensor \oc 1 \arrow[d, "\int_{A,1}"] \\
    \oc A \arrow[r, "\oc\runit_A"] \arrow[u, "\runit_{\oc A}"]
      & \oc (A \times 1)
  \end{tikzcd}\quad
  \begin{tikzcd}
    \oc A \tensor \oc B \arrow[r, "\sym_{A,B}"] \arrow[d, "\int_{A,B}"']
      & \oc B \tensor \oc A \arrow[d, "\int_{B,A}"] \\
    \oc(A\times B) \arrow[r, "\oc\sym_{A,B}"]
      & \oc(B\times A)
  \end{tikzcd}
  \]

We first prove a small lemma, which gives us a simpler way to compute $\oc \sigma$ in the case that $\sigma$ is a morphism in $\C_s$.  

\begin{lemma}\label{ocStrict}
  Let $A,B$ be objects of $\C_s$ and let $\sigma$ be a morphism from $A$ to $B$ in $\C_s$.  Then the following diagram commutes:
  \[
    \begin{tikzcd}
      \oc A \arrow[r, "\alpha_A"] \arrow[d, "\oc \sigma"']
        & A \sequoid \oc A \arrow[r, "\sigma\sequoid\id"]
          & B \sequoid \oc A \arrow[d, "\id\sequoid\oc\sigma"] \\
      \oc B \arrow[rr, "\alpha_B"]
        &
          & B \sequoid \oc B
    \end{tikzcd}
    \]
\end{lemma}

\begin{proof}
  By the definition of $\oc \sigma$, we have the following commutative diagram:
  \[
    \begin{tikzcd}
      \oc A \arrow[r, "\mu_A"] \arrow[d, "\oc \sigma"']
        & \oc A \tensor \oc A \arrow[r, "\der\tensor\id"]
          & A \tensor \oc A \arrow[r, "\sigma\tensor\id"]
            & B \tensor \oc A \arrow[r, "\wk"]
              & B \sequoid \oc A \arrow[d, "\id_B\sequoid\oc\sigma"] \\
      \oc B \arrow[rrrr, "\alpha_B"]
        &
          &
            &
              & B \sequoid \oc B
    \end{tikzcd}
    \]
  Therefore, it will suffice to show that the following diagram (solid lines) commutes:
  \[
    \begin{tikzcd}
      \oc A \arrow[r, "\alpha_A"] \arrow[d, "\mu_A"']
        & A \sequoid \oc A \arrow[r, "\sigma\sequoid\id"]
          & B \sequoid \oc A \\
      \oc A \tensor \oc A \arrow[r, "\der\tensor\id"]
        & A \tensor \oc A \arrow[u, dotted, "\wk"] \arrow[r, "\sigma\tensor\id"]
          & B \tensor \oc A \arrow[u, "\wk"]
    \end{tikzcd}
    \]
  The left hand square commutes by Lemma \ref{aFormulaForAlpha}.  The right hand square commutes because $\wk$ is a natural transformation.
\end{proof}

To show that the first coherence diagram commutes, we define a composite $\eta_{A,B,C}$:
\settowidth{\arrow}{\scriptsize$\id_{A\sequoid(\oc A\tensor\oc B)}\times (\id_B\sequoid\sym_{\oc B,\oc A})$}
\begin{gather*}
  (\oc A \tensor \oc B) \tensor \oc C \constantwidthxrightarrow{\langle\id,\; \sym\rangle} ((\oc A \tensor \oc B) \tensor \oc C) \times (\oc C \tensor (\oc A \tensor \oc B)) \\
  \hspace{6pt}\cdots \constantwidthxrightarrow{((\comp{\dist\inv}{\kappa_{A,B}})\tensor\id)\times(\alpha_C\tensor\id)} (((A \times B) \sequoid (\oc A \tensor \oc B)) \tensor \oc C) \times ((C \sequoid \oc C) \tensor (\oc A \tensor \oc B)) \\
  \hspace{6pt}\cdots \constantwidthxrightarrow{\wk\times\wk} (((A \times B) \sequoid (\oc A \tensor \oc B)) \sequoid \oc C) \times ((C \sequoid \oc C) \sequoid (\oc A \tensor \oc B)) \\
  \hspace{6pt}\cdots \constantwidthxrightarrow{\passoc\inv\times\passoc\inv} ((A \times B) \sequoid ((\oc A \tensor \oc B) \tensor \oc C)) \times (C \sequoid (\oc C \tensor (\oc A \tensor \oc B))) \\
  \hspace{6pt}\cdots \constantwidthxrightarrow{\id\times(\id\sequoid\sym)} ((A \times B) \sequoid ((\oc A \tensor \oc B) \tensor \oc C)) \times (C \sequoid ((\oc A \tensor \oc B) \tensor \oc C))
\end{gather*}

Observe the similarity between the definition of $\eta_{A,B,C}$ and that of $\kappa_{A\times B,C}$.  Indeed, it may be easily verified that the following diagram commutes, using the definition of $\int_{A,B}$ as the anamorphism for $\comp{\dist\inv}{\kappa_{A,B}}$ and the fact that $\wk$, $\passoc$ and $\sym$ are natural transformations:
\begin{equation}\label{kreuz}
  \begin{tikzcd}
    (\oc A \tensor \oc B) \tensor \oc C \arrow[r, "\eta_{A,B,C}"] \arrow[d, "\int_{A,B}\tensor\id"']
      & ((A \times B) \sequoid ((\oc A \tensor \oc B) \tensor \oc C)) \times (C \sequoid ((\oc A \tensor \oc B) \tensor \oc C)) \arrow[d, "(\id\sequoid(\int\tensor\id))\times (\id\sequoid(\int\tensor\id))"] \\
    \oc (A \times B) \tensor \oc C \arrow[r, "\kappa_{A\times B,C}"]
      & ((A \times B) \sequoid (\oc (A \times B) \tensor \oc C)) \times (C \sequoid (\oc (A \times B) \tensor \oc C))
  \end{tikzcd}\tag{\kreuz}
\end{equation}
Then we get the commutative diagram in Figure \ref{CoherenceOneWay}, which tells us that one of the two paths round the coherence diagram is the anamorphism of $\comp{\assoc_\times\sequoid\id}{\comp{\dist\inv}{\eta_{A,B,C}}}$.  We now show that the other path round the coherence diagram is the anamorphism of the same thing, which will prove that they are equal.  

\begin{SidewaysFigure}
  \[
    \begin{tikzcd}[ampersand replacement=\&]
      (\oc A \tensor \oc B) \tensor \oc C \arrow[r, "\comp{\dist\inv}{\eta_{A,B,C}}"] \arrow[d, "\int_{A,B}\tensor\id"'] \arrow[phantom, dr, "\textbf{a}"]
        \&[24pt] ((A \times B) \times C) \sequoid ((\oc A \tensor \oc B) \tensor \oc C) \arrow[r, "\assoc_\times\sequoid\id"] \arrow[d, "\id\sequoid(\int\tensor\id)"] \arrow[phantom, dr, "\textbf{b}"]
          \&[19pt] (A \times (B\times C)) \sequoid ((\oc A \tensor \oc B) \tensor \oc C) \arrow[d, "\id\sequoid(\int\tensor\id)"] \\
      \oc (A \times B) \tensor \oc C \arrow[r, "\comp{\dist\inv}{\kappa_{A\times B,C}}"] \arrow[d, "\int_{A\times B,C}"'] \arrow[phantom, dr, "\textbf{d}"]
        \& ((A \times B) \times C) \sequoid (\oc(A\times B)\tensor \oc C) \arrow[r, "\assoc_\times\sequoid\id"] \arrow[d, "\id\sequoid\int"] \arrow[phantom, dr, "\textbf{c}"]
          \& (A \times (B\times C)) \sequoid (\oc(A\times B)\tensor\oc C) \arrow[d, "\id\sequoid\int"] \\
      \oc((A\times B)\times C) \arrow[r, "\alpha_{(A\times B)\times C}"] \arrow[d, "\oc \assoc_\times"'] \arrow[phantom, drr, "\textbf{e}"]
        \& ((A\times B) \times C) \sequoid \oc ((A\times B)\times C) \arrow[r, "\assoc_\times\sequoid\id"]
          \& (A \times (B\times C)) \sequoid \oc (A \times(B \times C)) \arrow[d, "\id\sequoid\oc \assoc_\times"] \\
      \oc(A \times (B\times C)) \arrow[rr, "\alpha_{A\times(B\times C)}"]
        \&
          \& (A \times (B \times C)) \sequoid \oc (A \times (B \times C))
    \end{tikzcd}
    \]
  \caption{\textbf{a} commutes by Diagram \eqref{kreuz}, plus the fact that $\dist$ is a natural transformation.  \\[\baselineskip]
  \textbf{b} and \textbf{c} commute because $\assoc_\times$ is a natural transformation.  \\[\baselineskip]
  \textbf{d} commutes by the definition of $\int_{A\times B,C}$.  \\[\baselineskip]
  \textbf{e} commutes by Lemma \ref{ocStrict}.}\label{CoherenceOneWay}
\end{SidewaysFigure}

For this, we define a composite $\tilde{\eta}_{A,B,C}$:
\settowidth{\arrow}{\scriptsize$\id_{A\sequoid(\oc A\tensor\oc B)}\times (\id_B\sequoid\sym_{\oc B,\oc A})$}
\begin{gather*}
  \oc A \tensor (\oc B \tensor \oc C) \constantwidthxrightarrow{\langle\sym,\id\rangle} (\oc A \tensor (\oc B \tensor \oc C)) \times ((\oc B \tensor \oc C) \tensor \oc A) \\
  \hspace{6pt}\cdots \constantwidthxrightarrow{(\alpha_A\tensor\id)\times((\comp{\dist\inv}{\kappa_{B,C}})\tensor\id)} ((A \sequoid \oc A) \tensor (\oc B \tensor \oc C)) \times (((B \times C) \sequoid (\oc B \tensor \oc C)) \tensor \oc A) \\
  \hspace{6pt}\cdots \constantwidthxrightarrow{\wk\times\wk} ((A \sequoid \oc A) \sequoid (\oc B \tensor \oc C)) \times (((B \times C) \sequoid (\oc B \tensor \oc C)) \sequoid \oc A) \\
  \hspace{6pt}\cdots \constantwidthxrightarrow{\passoc\inv\times\passoc\inv} (A \sequoid (\oc A \tensor (\oc B \tensor \oc C))) \times ((B \times C) \sequoid ((\oc B \tensor \oc C) \tensor \oc A)) \\
  \hspace{6pt}\cdots \constantwidthxrightarrow{\id\times(\id\sequoid\sym)} (A \sequoid (\oc A \tensor (\oc B \tensor \oc C))) \times ((B \times C) \sequoid (A \tensor (\oc B \tensor \oc C))) \\
\end{gather*}

We get a commutative diagram:
\[
  \begin{tikzcd}
    (\oc A \tensor \oc B) \tensor \oc C \arrow[r, "\comp{\assoc_\times\sequoid\id}{\comp{\dist\inv}{\eta_{A,B,C}}}"] \arrow[d, "\assoc_{A,B,C}"']
      &[78pt] (A \times (B \times C)) \sequoid ((\oc A \tensor \oc B) \tensor \oc C) \arrow[d, "\id\sequoid\assoc"] \\
    \oc A \tensor (\oc B \tensor \oc C) \arrow[r, "\comp{\assoc_\times\sequoid\id}{\comp{\dist\inv}{\tilde{\eta}_{A,B,C}}}"] \arrow[d, "\id\tensor\int"']
      & (A \times (B \times C)) \sequoid (\oc A \tensor (\oc B \tensor \oc C)) \arrow[d, "\id\sequoid(\id\tensor\int)"] \\
    \oc A \tensor \oc(B \times C) \arrow[r, "\comp{\assoc_\times\sequoid\id}{\comp{\dist\inv}{\kappa_{A,B\times C}}}"] \arrow[d, "\int_{A,B\times C}"']
      & (A \times (B \times C)) \sequoid (\oc A \tensor \oc (B\times C)) \arrow[d, "\id\sequoid\int"] \\
    \oc (A \times (B \times C)) \arrow[r, "\alpha_{A\times(B\times C)}"]
      & (A \times (B \times C)) \sequoid \oc (A \times (B \times C))
  \end{tikzcd}
  \]

Commutativity of the top square is a long and fairly unenlightening exercise in the coherence conditions for symmetric monoidal categories and for sequoidal categories.  Commutativity of the middle and bottom squares are by similar arguments to the ones in Figure \ref{CoherenceOneWay}.  Therefore, the two branches of the coherence diagram are anamorphisms for the same thing, and so they are equal.  

The proofs for the other three coherence diagrams are similar.  

\subsection{Proof of Proposition \ref{itsMu}}

\begingroup
\def\thetheorem{\ref{itsMu}}
\begin{proposition}
  $\CCom(\oc)\left(A\xrightarrow{\Delta}A\times A\right)$ has comultiplication given by $\mu_A\from\oc A \to \oc A\tensor\oc A$ and counit given by the unique morphism $\eta_A\from \oc A \to I$.
\end{proposition}
\addtocounter{theorem}{-1}
\endgroup

By the definition of $\CCom$, the comultiplication in $\CCom(\oc)\left(A\xrightarrow{\Delta}A\times A\right)$ is given by the composite:
\[
  \oc A \xrightarrow{\oc\Delta} \oc (A \times A) \xrightarrow{\int_{A,A}\inv} \oc A \times \oc A
  \]
By Lemma \ref{ocStrict}, the morphism $\oc\Delta$ is equal to the morphism $\sigma_A$ defined above.  So this composite is equal to $\comp{\int_{A,A}}{\sigma_A}=\mu_A$.  

The counit is a morphism $\oc A\to I$, so by uniqueness it must be equal to $\eta_A$.  

\subsection{Proof of Theorem \ref{Coalgebra__CoCoCo}}

\begingroup
\def\thetheorem{\ref{Coalgebra__CoCoCo}}
\begin{theorem}
  Let $(\C,\C_s,J,\wk)$ be a sequoidal category satisfying all the conditions from Theorem \ref{StrongMonoidalFunctor}.  Let $A$ be an object of $\C$ (equivalently, of $\C_s$).  Then $\oc A$, together with the comultiplication $\mu_A$ and counit $\eta_A$, is the cofree commutative comonoid over $A$.
\end{theorem}
\addtocounter{theorem}{-1}
\endgroup

We know from Proposition \ref{itsMu} that the $(\oc A,\mu_A,\eta_A)$ is indeed a commutative comonoid.  Now let $\delta\from B\to B\tensor B$ be a commutative comonoid in $\C$ and let $f\from B\to A$ be a morphism.  We need to show that there is a unique morphism $f^\dag\from B \to \oc A$ such that the following diagram commutes:
\[
  \begin{tikzcd}
      & B \arrow[r, "\delta"] \arrow[d, "f^\dag"] \arrow[dl, "f"']
        & B \tensor B \arrow[d, "f^\dag\tensor f^\dag"] \\
    A
      & \oc A \arrow[l, "\der_A"] \arrow[r, "\mu_A"]
        & \oc A \tensor \oc A
  \end{tikzcd}
  \]
We define the morphism $f^\dag$ to be the anamorphism of the composite:
\[
  B \xrightarrow{\delta} B \tensor B \xrightarrow{f\tensor\id B} A \tensor B \xrightarrow{\wk_{A,B}} A \sequoid B
  \]
We first claim that it makes the given diagram commute, starting with the square on the left.  We show that $\comp{\int_{A,A}}{\comp{\mu_A}{f^\dag}}=\comp{\int_{A,A}}{\comp{(f^\dag\tensor f^\dag)}{\delta}}\from B \to \oc (A\times A)$, by showing that both morphisms are anamorphisms for the composite
\[
  B \xrightarrow{\delta} B \tensor B \xrightarrow{f\tensor\id_B} A \tensor B \xrightarrow{\wk_{A,B}} A \sequoid B \xrightarrow{\Delta} (A \sequoid B) \times (A \sequoid B) \xrightarrow{\dist_{A,B}\inv} (A \times A) \sequoid B
  \]
Since $\int_{A,A}$ is an isomorphism, this will prove that the square on the right commutes.  

For the first case, the diagram in Figure \ref{transferReportDiagram1} proves that $\comp{\int_{A,A}}{\comp{\mu_A}{f^\dag}}=\comp{\sigma_A}{f^\dag}$ is the anamorphisms for that composite.  For the other case, taking the product of the diagrams in Figure \ref{transferReportDiagram2} gives rise to the diagram in Figure \ref{transferReportDiagram3}, which proves that $\comp{\int_{A,A}}{\comp{(f^\dag\tensor f^\dag)}{\delta}}$ is the anamorphism for the same composite, which completes the proof that the right hand square commutes.  
\begin{SidewaysFigure}
  \[
    \begin{tikzcd}[cramped, ampersand replacement=\&]
      B \arrow[r, "\delta"] \arrow[d, "f^\dag"'] \arrow[phantom, drrr, "\textbf{a}"]
        \& B \tensor B \arrow[r, "f\tensor\id_B"]
          \& A \tensor B \arrow[r, "\wk_{A,B}"]
            \& A \sequoid B \arrow[r, "\Delta"] \arrow[d, "\id\sequoid f^\dag"] \arrow[phantom, dr, "\textbf{b}" xshift=0.6em]
              \& (A \sequoid B) \times (A \sequoid B) \arrow[r, "\dist_{A,B}\inv"] \arrow[d, "(\id\sequoid f^\dag)\times(\id\sequoid f^\dag)"] \arrow[phantom, dr, "\textbf{c}" xshift=2.0em]
                \& (A \times A) \sequoid B \arrow[d, "\id\sequoid f^\dag"] \\
      \oc A \arrow[rrr, "\alpha_A"] \arrow[d, "\sigma_A"'] \arrow[phantom, drrrrr, "\textbf{d}"]
        \&
          \&
            \& A \sequoid \oc A \arrow[r, "\Delta"]
              \& (A \sequoid \oc A) \times (A \sequoid \oc A) \arrow[r, "\dist_{A,\oc A}\inv"]
                \& (A \times A) \sequoid \oc A \arrow[d, "\id\sequoid\sigma_A"] \\
      \oc (A\times A) \arrow[rrrrr, "\alpha_{A\times A}"]
        \&
          \&
            \&
              \&
                \& (A \times A) \sequoid \oc (A\times A)
    \end{tikzcd}
    \]
  \caption{\textbf{a} commutes by the definition of $f^\dag$.\\[\baselineskip]
    \textbf{b} commutes because $\Delta$ is a natural transformation.  \textbf{c} commutes because $\dist$ is a natural transformation. \\[\baselineskip]
    \textbf{d} commutes be the definition of $\sigma_A$.
  }\label{transferReportDiagram1}
\end{SidewaysFigure}

\begin{SidewaysFigure}
  \[
    \begin{tikzcd}[ampersand replacement=\&]
      B \arrow[r, "\delta"] \arrow[dd, "\delta"'] \arrow[ddr, phantom, "\textbf{a}"]
        \& B\tensor B \arrow[r, "f\tensor \id_B"] \arrow[d, "\id_B\tensor\delta"]
          \& A \tensor B \arrow[rrr, "\wk"] \arrow[d, "\id_A\tensor\delta"] \arrow[drrr, phantom, "\textbf{b}"]
            \&
              \&
                \& A\sequoid B \arrow[d, "\id_A\sequoid\delta"] \\
        \& B \tensor(B\tensor B) \arrow[r, "f\tensor\id_{B\tensor B}"] \arrow[d, "{\assoc_{B,B,B}\inv}"] \arrow[phantom, dr, "\textbf{d}" xshift=0.3em]
          \& A \tensor (B \tensor B) \arrow[rrr, "\wk"] \arrow[drr, phantom, "\textbf{g}"] \arrow[d, "\assoc_{A,B,B}\inv"]
            \&
              \&
                \& A \sequoid (B\tensor B) \arrow[dd, "\id_A\sequoid(f^\dagger \tensor f^\dagger)"] \arrow[dl, "{\passoc_{A,B,B}}"] \arrow[ddl, phantom, "\textbf{e}"] \\
      B \tensor B \arrow[r, "\delta\tensor \id_B"] \arrow[d, "f^\dagger \tensor f^\dagger"'] \arrow[drrr, phantom, "\textbf{f}" yshift=-0.3em]
        \& (B \tensor B) \tensor B \arrow[r, "(f\tensor\id_B)\tensor\id_B"' yshift=-0.3em]
          \& (A \tensor B) \tensor B \arrow[r, "\wk\tensor\id_B"]
            \& (A \sequoid B) \tensor B \arrow[r, "\wk"] \arrow[d, "(\id\sequoid f^\dagger)\tensor f^\dagger"'] \arrow[dr, phantom, "\textbf{c}"]
              \& (A\sequoid B)\sequoid B \arrow[d, "(\id\sequoid f^\dagger)\sequoid f^\dagger"]
                \& \\
      \oc A \tensor \oc A \arrow[rrr, "\alpha\tensor\id_{\oc A}"']
        \&
          \&
            \& (A \sequoid \oc A)\tensor \oc A \arrow[r, "\wk"]
              \& (A \sequoid \oc A) \sequoid \oc A \arrow[r, "\passoc\inv"]
                \& A \sequoid (\oc A \tensor \oc A)
    \end{tikzcd}
  \]
\[
  \begin{tikzcd}[ampersand replacement=\&]
    \& B \arrow[r, "\delta"] \arrow[d, "\delta"] \arrow[dl, "\delta"'] \arrow[phantom, ddl, "\textbf{h}" near start] \arrow[phantom, ddrrr, "\textbf{j}"]
        \& B \tensor B \arrow[r, "f\tensor\id_B"]
          \& A \tensor B \arrow[r, "\wk"]
            \& A \sequoid B \arrow[d, "\id\sequoid\delta"] \\
    B \tensor B \arrow[r, "\sym"] \arrow[d, "f^\dag\tensor f^\dag"'] \arrow[phantom, dr, "\textbf{i}"]
      \& B \tensor B \arrow[d, "f^\dag\tensor f^\dag"]
        \&
          \&
            \& A \sequoid (B \tensor B) \arrow[d, "\id\sequoid(f^\dag\tensor f^\dag)"] \\
    \oc A \tensor \oc A \arrow[r, "\sym"]
      \& \oc A \tensor \oc A \arrow[r, "\alpha_A \tensor \id_A"']
        \& (A \sequoid \oc A) \tensor \oc A \arrow[r, "\wk"]
          \& (A \sequoid \oc A) \sequoid \oc A \arrow[r, "\passoc\inv"]
            \& A \sequoid (\oc A \tensor \oc A)
  \end{tikzcd}
  \]
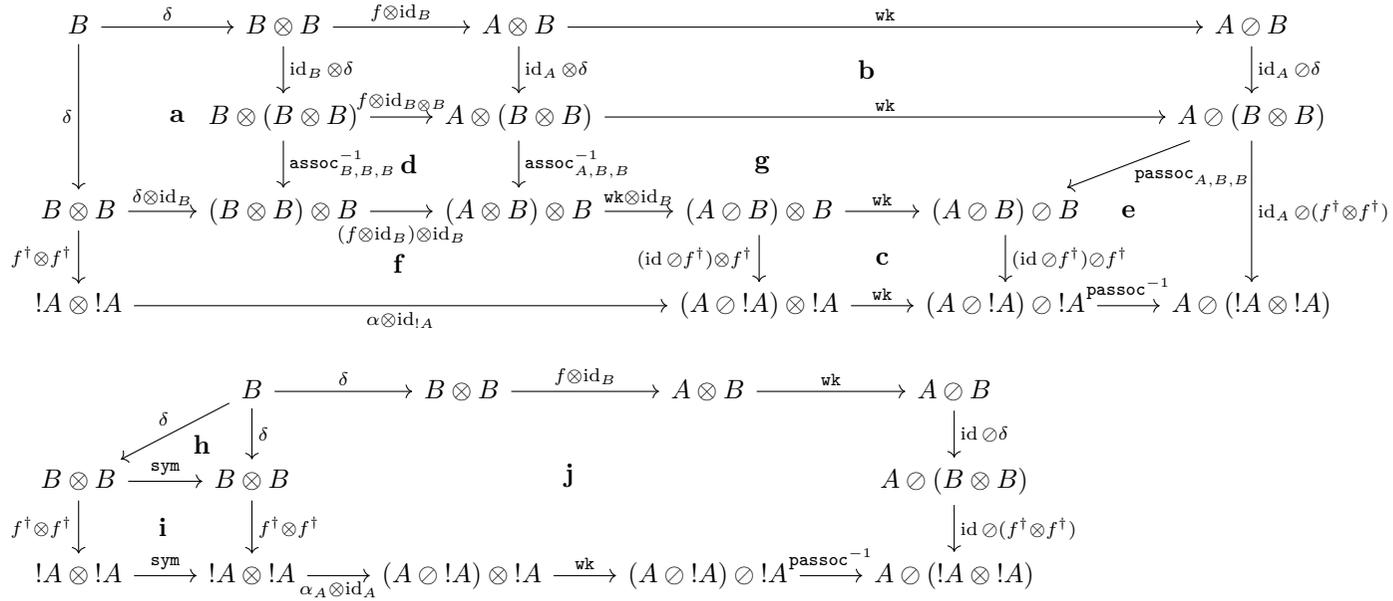
  \caption{\textbf{a} commutes because the comultiplication $\delta$ is associative.\\[\baselineskip]
    \textbf{b} and \textbf{c} commute because $\wk$ is a natural transformation.  \textbf{d} and \textbf{e} commute because $\assoc$ and $\passoc$ are natural transformations.  \\[\baselineskip]
    \textbf{f} is the tensor product of two diagrams: one is the definition of $f^\dag$ and the other one obviously commutes.\\[\baselineskip]
    \textbf{g} is one of the coherence diagrams for $\wk$.\\[\baselineskip]
    \textbf{h} commutes because the comultiplication $\delta$ is commutative. \\[\baselineskip]
    \textbf{i} commutes because $\sym$ is a natural transformation.\\[\baselineskip]
    \textbf{j} commutes by the first diagram.
  }\label{transferReportDiagram2}
\end{SidewaysFigure}

\begin{SidewaysFigure}
  \[
    \begin{tikzcd}[ampersand replacement=\&]
      B \arrow[r, "\delta"] \arrow[d, "\delta"'] \arrow[phantom, ddrrrr, "\textbf{a}"]
        \& B \tensor B \arrow[r, "f\tensor\id_B"]
          \& A \tensor B \arrow[r, "\wk"]
            \& A \sequoid B \arrow[r, "\Delta"]
              \& (A \sequoid B) \times (A \sequoid B) \arrow[r, "\dist_{A,B}\inv"] \arrow[d, "(\id_A\sequoid\delta)\times(\id_A\sequoid\delta)"'] \arrow[phantom, ddr, "\textbf{b}" xshift=1.2em]
                \& (A \times A) \sequoid B \arrow[d, "\id_{A\times A}\sequoid\delta"] \\
      B \tensor B \arrow[d, "f^\dag\tensor f^\dag"']
        \&
          \&
            \&
              \& (A \sequoid (B \tensor B)) \times (A \sequoid (B \tensor B)) \arrow[d, "(\id_A \sequoid (f^\dag\tensor f^\dag))\times(\id_A \sequoid (f^\dag\tensor f^\dag))"']
                \& (A \times A) \sequoid (B \tensor B) \arrow[d, "\id_{A\times A} \sequoid (f^\dag\tensor f^\dag)"] \\
      \oc A \tensor \oc A \arrow[rrrr, "\sigma_A"] \arrow[d, "\int_{A,A}"'] \arrow[phantom, drrrrr, "\textbf{c}"]
        \&
          \&
            \&
              \& (A \sequoid (\oc A \tensor \oc A)) \times (A \sequoid (\oc A \tensor \oc A)) \arrow[r, "\dist_{A,\oc A \tensor \oc A}\inv" yshift=0.3em]
                \& (A \times A) \sequoid (\oc A \tensor \oc A) \arrow[d, "\id_{A\times A}\sequoid\int_{A,A}"] \\
      \oc (A\times A) \arrow[rrrrr, "\alpha_{A\times A}"]
        \&
          \&
            \&
              \&
                \& (A \times A) \sequoid \oc(A\times A)
    \end{tikzcd}
    \]
  \caption{\textbf{a} is the product of the diagrams in Figure \ref{transferReportDiagram2}.\\[\baselineskip]
    \textbf{b} commutes because $\dist$ is a natural transformation.\\[\baselineskip]
    \textbf{c} is the definition of $\int_{A,A}$
  }\label{transferReportDiagram3}
\end{SidewaysFigure}

Now the following diagrams show that the triangle on the left commutes:
\[
  \begin{tikzcd}
    B \arrow[r, "\delta"] \arrow[d, "f^\dag"']
      & B \tensor B \arrow[r, "f\tensor\id_B"]
        & A \tensor B \arrow[r, "\wk"]
          & A \sequoid B \arrow[dr, "\id_A\sequoid*"] \arrow[d, "\id_A\sequoid f^\dag"']
            & 
              & \\
    \oc A \arrow[rrr, "\alpha_A"']
      &
        &
          & A \sequoid \oc A \arrow[r, "\id_A\sequoid*"']
            & A \sequoid I \arrow[r, "\run_A"]
              & A
  \end{tikzcd}
  \]
\[
  \begin{tikzcd}
    B \tensor B \arrow[r, "f\tensor\id"] \arrow[d, "\id_B\tensor*"]
      & A \tensor B \arrow[r, "\wk"] \arrow[d, "\id\tensor*"]
        & A \sequoid B \arrow[d, "\id\sequoid*"] \\
    B \tensor I \arrow[r, "f\tensor\id_I"]
      & A \tensor I \arrow[r, "\wk"]
        & A \sequoid I \arrow[dl, "\run_A"] \\
    B \arrow[uu, bend left=100, "\delta"] \arrow[u, "\runit_B"'] \arrow[r, "f"]
      & A \arrow[u, "\runit_A"]
        &
  \end{tikzcd}
  \]
In the first diagram, the morphism along the bottom is $\der_A$, by definition.  The second diagram shows that the morphism along the top of the first diagram is equal to $f$.  The triangle at the bottom right of that diagram is one of the coherence conditions for $\wk$, while the semicircle at the left commutes because the comultiplication $\delta$ is unital (with unit $*\from B\to I$).  

Lastly, we show uniqueness.  Suppose that $g\from B\to\oc A$ makes the diagram commute:
\[
  \begin{tikzcd}
      & B \arrow[r, "\delta"] \arrow[d, "g"] \arrow[dl, "f"']
        & B \tensor B \arrow[d, "g\tensor g"] \\
    A
      & \oc A \arrow[l, "\der_A"] \arrow[r, "\mu_A"]
        & \oc A \tensor \oc A
  \end{tikzcd}
  \]
We may convert this into the following diagram:
\[
  \begin{tikzcd}
    B \arrow[r, "\delta"] \arrow[d, "g"']
      & B \tensor B \arrow[r, "f\tensor\id_B"] \arrow[d, "g\tensor g"]
        & A \tensor B \arrow[r, "\wk"] \arrow[d, "\id_A\tensor g"]
          & A \sequoid B \arrow[d, "\id_A\sequoid g"] \\
    \oc A \arrow[r, "\mu_A"]
      & \oc A \tensor \oc A \arrow[r, "\der_A\tensor\id"]
        & A \tensor \oc A \arrow[r, "\wk"]
          & A \sequoid \oc A
  \end{tikzcd}
  \]
Here, the left hand square is taken straight from the previous diagram, while the middle square is the tensor product of the left hand triangle with a diagram that obviously commutes.  The right hand square commutes because $\wk$ is a natural transformation.  

By Lemma \ref{aFormulaForAlpha}, the morphism along the bottom is equal to $\alpha_A$ and therefore $g$ is the anamorphism for the morphism along the top; i.e., $g=f^\dag$.  

\end{document}